\documentclass[11pt,letterpaper,dvipsnames]{article}
\usepackage[margin=1in]{geometry}
\usepackage[titles]{tocloft}
 
\newcommand{\ignore}[1]{}
 
\usepackage[x11names, rgb]{xcolor}
\usepackage[utf8]{inputenc}
\usepackage{fullpage}
\usepackage{stmaryrd}
\usepackage{tikz-cd}
\addtocontents{toc}{\protect\setcounter{tocdepth}{2}}%

\usepackage{bm}
\usepackage{todonotes}
\usepackage{lipsum}
\usepackage{setspace}
\usepackage{scrextend}
\usepackage{ragged2e}

\usepackage[hidelinks, pagebackref]{hyperref}
\usepackage{xcolor}
\hypersetup{
    colorlinks,
    linkcolor={red!40!black},
    citecolor={green!50!black},
    urlcolor={blue!40!black}
}
\usepackage{microtype}%if unwanted, comment out or use option "draft"
\usepackage{textcomp}
\usepackage{amsmath, amsfonts, amssymb, amsthm,complexity}
\usepackage{mathtools, mathdots,nccmath,thmtools}
\usepackage[sc]{mathpazo}
\usepackage{tgpagella}
\usepackage{mhsetup}
\usepackage{enumerate}
\usepackage[linesnumbered,algoruled,boxed,lined]{algorithm2e}

\newtheorem{theorem}{Theorem}
\newtheorem*{theorem*}{Theorem}

\newtheorem{lemma}[theorem]{Lemma}
\newtheorem{definition}{Definition}

\newtheorem{claim}[theorem]{Claim}

\newtheorem{observation}{Observation}
\newtheorem*{observation*}{Observation}
\newtheorem{problem}{Problem}

\theoremstyle{remark}

\newcommand{\F}{\mathbb{F}}
\newcommand{\N}{\mathbb{N}}

\newcommand{\Q}{\mathbb{Q}}
\newcommand{\Z}{\mathbb{Z}}

\newcommand{\calS}{\mathcal{S}}

\newcommand{\cf}{\textup{{coef}}}

\newcommand{\x}{{\boldsymbol x}}
\newcommand{\y}{{\boldsymbol y}}

\newcommand{\e}{{\boldsymbol e}}

\newcommand{\calI}{\mathcal{I}}

\newcommand{\calP}{\mathcal{P}}

\newcommand{\m}{\boldsymbol m}

\newcommand{\floor}[1]{\lfloor #1 \rfloor}

\newcommand{\ceil}[1]{\lceil #1 \rceil}

\newcommand{\sssum}[1]{#1-\mathsf{SimulSubsetSum}}

\newcommand{\ssum}{\mathsf{SSUM}}
\newcommand{\hsssum}[1]{\mathsf{Hamming}-#1-\mathsf{SSSUM}}
\newcommand{\hsubssum}[1]{\mathsf{Hamming}-#1-\mathsf{SUBSSUM}}
\newcommand{\ubssum}{\mathsf{UBSSUM}}
\newcommand{\ksubssum}[1]{#1-\mathsf{SUBSSUM}}
\newcommand{\ksssum}[1]{#1-\mathsf{SSSUM}}
\newcommand{\usssum}{\mathsf{uSSSUM}}

\newcommand{\potential}{\mathbb{P}}

\newcommand{\pprod}{\mathsf{Subset~Product}}
\newcommand{\simulsub}{\mathsf{SimulSubsetSum}}

\newtheorem{remark}{Remarks}
\newtheorem*{remark*}{Remarks}

\newcommand{\OO}{\tilde{O}}

\title{Efficient reductions and algorithms for variants of Subset Sum~\thanks{A preliminary version of this work has been accepted to CALDAM 2022.}}

\author{Pranjal Dutta~\thanks{Chennai Mathematical Institute, India (and Visiting Research Fellow, Department of CSE, IIT Kanpur). Email:~~\texttt{pranjal@cmi.ac.in}.} \and Mahesh Sreekumar Rajasree~\thanks{Department of CSE, IIT Kanpur. Email:~~\texttt{mahesr@cse.iitk.ac.in}.}}
\date{}

\begin{document}

\maketitle

\begin{abstract}
Given $(a_1, \dots, a_n, t) \in \Z_{\geq 0}^{n + 1}$, the Subset Sum problem ($\ssum$) is to decide whether there exists $S \subseteq [n]$ such that $\sum_{i \in S} a_i = t$. There is a close variant of the $\ssum$, called $\pprod$. Given positive integers $a_1, \hdots, a_n$ and a target integer $t$, the $\pprod$ problem asks to
determine whether there exists a subset $S \subseteq [n]$ such that  $\prod_{i \in S} a_i=t$. There is a pseudopolynomial time dynamic programming algorithm, due to Bellman (1957) which solves the $\ssum$ and $\pprod$ in $O(nt)$ time and $O(t)$ space.

In the first part, we present {\em search} algorithms for variants of the Subset Sum problem. Our algorithms are parameterized by $k$, which is a given upper bound on the number of realisable sets (i.e.,~number of solutions, summing exactly $t$). We show that $\ssum$ with a unique solution is already $\NP$-hard, under randomized reduction. This makes the regime of parametrized algorithms, in terms of $k$, very interesting.

Subsequently, we present an $\tilde{O}(k\cdot (n+t))$ time deterministic algorithm, which finds the hamming weight of all the realisable sets for a subset sum instance. We also give a $\poly(knt)$-time and $O(\log(knt))$-space deterministic algorithm that finds all the realisable sets for a subset sum instance.

In the latter part, we present a simple and elegant randomized algorithm for $\pprod$ in $\OO(n + t^{o(1)})$ expected-time. Moreover, we also present a $\poly(nt)$ time and $O(\log^2 (nt))$ space deterministic algorithm for the same. We study these problems in the unbounded setting as well. Our algorithms use multivariate FFT, power series and number-theoretic techniques, introduced by Jin and Wu (SOSA'19) and Kane (2010).\\

% In fact, we solve a more general problem called the $\simulsub$. This problem was introduced by Kane 2010. Given $k$ instances of \textsf{Subset~Sum}, it asks to decide whether there is a `common' solution to all the instances. Kane gave a logspace algorithm for this problem. We show a polynomial-time  reduction from  $\pprod$ to $\simulsub$ and also give efficient algorithm for the latter. , with some additional algebraic tools and observations. %Our algorithms are optimal (in time, respectively~space complexity), for $k=O(1)$.

\smallskip
{\bf \noindent Keywords.}~subset sum, subset product, power series, isolation lemma, hamming weight, interpolation, logspace, Newton's identities, simultaneous, FFT, pseudo-prime-factor
\end{abstract}

\newpage

\tableofcontents

\newpage 
\section{Introduction}

\vspace{-2mm}
The  \textsf{Subset Sum} problem ($\ssum$) is a well-known $\NP$-complete problem \cite[p.~226]{lewis1983computers}, where given $(a_1, \dots, a_n, t) \in \Z_{\geq 0}^{n + 1}$, the problem is to decide whether there exists $S \subseteq [n]$ such that $\sum_{i \in S}\,a_i = t$. In the recent years, provable-secure cryptosystems based on $\ssum$ such as private-key encryption schemes \cite{lyubashevsky2010public}, tag-based encryption schemes \cite{faust2016chosen}, etc have been proposed. There are numerous improvements made in the algorithms that solve the~$\ssum$~ problem in both the classical \cite{bringmann2021near,bringmann2017near,jin2018simple,jin2021fast,esser2020low} and quantum world \cite{bernstein2013quantum,helm2018subset,li2019improved}. One of the first algorithms was due to Bellman \cite{bellman57} who gave a $O(nt$) time ({\em pseudo-polynomial} time) algorithm which requires $\Omega(t)$ space. In this paper, we give efficient algorithms for interesting variants of subset sum.

\subsection{Variants of Subset Sum}

To begin with, one can ask for a {\em search} version of the subset sum problem, i.e.,~to output all the solutions. Since there can be {\em exponentially} many solutions, it could take $\exp(n)$-time (and space), to output them. This motivates our first problem~defined below.
\begin{problem}[$\ksssum{k}$] \label{prob1:k-sssum}
Given $(a_1, \dots, a_n, t) \in \Z_{\geq 0}^{n+1}$, the $k$-solution $\ssum$ $(\ksssum{k})$ problem asks to output all $S \subseteq [n]$ such that $\sum_{i \in S}\,a_i = t$ provided with the guarantee that the number of such subsets is at most $k$.
\end{problem}

{\noindent $\blacktriangleright$ \textsf{Remark}.}~~We denote $\ksssum{1}$ as unique Subset Sum problem ($\usssum$). In ~\href{https://cstheory.stackexchange.com/questions/42502/subset-sum-problem-with-at-most-one-solution-for-any-target}{stackexchange}, a more restricted version was asked where it was assumed that $k=1$, for {\em any} realizable $t$. Here we just want $k=1$ for some {\em fixed} target value $t$ and we do not assume anything for any other value $t'$.

\smallskip
Now, we consider a different restricted version of the $\ksssum{k}$, where we demand to output only the hamming weights of the $k$-solutions (we call it $\hsssum{k}$, for definition see \autoref{prob2:ham-k}). By hamming weight of a solution, we mean the number of $a_i$'s in the solution set (which sums up to exactly $t$). In other words, if $\vec{a} \cdot \vec{v}=t$, where $\vec{a} = (a_1, \hdots, a_n)$ and $\vec{v} \in \{0,1\}^n$, we want $|v|_1$, the $\ell_1$-norm of the solution vector. %This is a restricted version of \cite{shallue2012division} where only $w$, the hamming-weight is fixed

\begin{problem}[$\hsssum{k}$] \label{prob2:ham-k}
Given an instance of the $\ksssum{k}$, say~$(a_1, \dots, a_n, t) \in \Z_{\ge 0}^{n+1}$, with the promise that there are at most $k$-many $S \subseteq [n]$ such that $\sum_{i \in S}a_i = t$, $\hsssum{k}$ asks to output all the hamming weights (i.e., $|S|$) of the solutions.
\end{problem}

It is obvious that solving $\ksssum{k}$ solves \autoref{prob2:ham-k}. Importantly, the decision problem, namely the \textsf{HWSSUM} is already $\NP$-hard. 

\paragraph{\textsf{HWSSUM} and its $\NP$-hardness.} The~\textsf{HWSSUM}~problem is :~Given an instance $(a_1,\hdots, a_n,t,w) \in \Z_{\ge 0}^{n+2}$, decide whether there is a solution to the subset sum with hamming weight equal to $w$. Note that, there is a trivial Cook's reduction from the $\ssum$ to the \textsf{HWSSUM}: $\ssum$ decides `yes' to the instance $(a_1,\hdots,a_n,t)$ iff at least one of the following \textsf{HWSSUM} instances $(a_1,\hdots,a_n,t,i)$, for $\,i \in [n]$ decides `yes'. Therefore, the search-version of \textsf{HWSSUM}, the $\hsssum{k}$ problem, is already an interesting problem and worth investigating.

\bigskip 
\noindent In parallel, we also study a well-known variant of the subset sum, called $\pprod$.

\begin{problem}[$\pprod$] \label{prob2:partition-product}
Given $(a_{1}, \dots, a_{n}, t)$ $\in \Z_{\geq 1}^{n+1}$, the $\pprod$ problem asks to decide whether there exists an $S \subseteq [n]$ such that $\prod_{i \in S} a_{i} = t$.
\end{problem}

\noindent $\pprod$ is also known to be $\NP$-complete~\cite[p.~221]{garey1979computers}. Like $\ssum$, it has a trivial $O(nt$) time ({\em pseudo-polynomial} time) dynamic programming algorithm which requires $\Omega(t)$ space~\cite{bellman57}. 

$\pprod$ has been studied and applied in many different forms. For e.g.,~1) constructing a {\em smooth hash} (VSH) by Contini, Lenstra and Steinfeld~\cite{contini2006vsh}, 2) attack on the Naccache-Stern Knapsack (NSK) public key cryptosystem~\cite{draziotis2020product}. Similar problem has also been studied in optimization, in the form of product knapsack problem~\cite{pferschy2021approximating}, multiobjective knapsack problem~\cite{bazgan2009solving} and so on.

%\cite{bringmann2021near,bringmann2017near,jin2018simple,jin2021fast,esser2020low} and quantum world \cite{bernstein2013quantum,helm2018subset,li2019improved}. One of the first algorithms was due to Bellman \cite{bellman57} who gave a $O(nt$) time ({\em pseudo-polynomial} time) algorithm which requires $\Omega(t)$ space. One can ask for a {\em search} version of this problem, i.e.~to output all the solutions. Since there can be {\em exponentially} many solutions, it could take $\exp(n)$-time (and space), to output them. This motivates our first problem~defined below.

Next, we define a `seemingly' unrelated problem at first. It asks to decide whether there is a `common' solution to the given many instances of subset sum. This was first introduced by~\cite[Section~3.3]{kane2010unary} (but no formal name was given). In this work, we study this as in intermediate problem which plays a crucial role to study the $\pprod$ problem, see~\autoref{sec:rand-algo-pprod}.
\begin{problem}[\textsf{SimulSubsetSum}]
\label{prob1:k-ssssum}
Given subset sum instances $(a_{1j}, \dots, a_{nj}, t_{j})$ $\in \Z_{\geq 0}^{n+1}$, for $j \in [k]$, where $k$ is some parameter, the Simultaneous Subset Sum problem (in short, $\textsf{SimulSubsetSum}$) asks to decide whether there exists an $S \subseteq [n]$ such that $\sum_{i \in S} a_{ij} = t_j, \forall j \in [k]$.
\end{problem}
{\noindent $\blacktriangleright$~\textsf{Remarks}.}~~1. When $k$ is fixed parameter (independent of $n$), we call this $\sssum{k}$. There is a trivial  $O(n(t_1+1) \hdots (t_k+1))$ time deterministic algorithm for the $\simulsub$ problem with $k$ subset sum instances ($k$ not necessarily a constant); for details see~\autoref{sec:dynamic-simul}. 

\smallskip
2. It suffices to work with $t_j \ge 1, \forall j \in [k]$. To argue that, let us assume that $t_j = 0$ for some $j \in [k]$ and $I_j := \{i \in [n]\,|\, a_{ij} = 0\}$. Observe that if $\simulsub$ has a solution set $S \subseteq [n]$, then $S \subseteq I_j$. Therefore, for every $\ell \in [k]$, instead of looking at $(a_{1\ell}, \dots, a_{n\ell}, t_{\ell})$, it suffices to work with $\{a_{i,\ell}\,|\, i \in I_{j}\}$ with the target $t_{\ell}$. Thus, we can trivially ignore the $j^{th}$ $\ssum$ instance.

\paragraph{Hardness depends on $k$.} \autoref{prob1:k-ssssum} asks to solve a system of $k$-linear equations in $n$-variables with $0/1$ constraints on the variables in a linear algebraic way. If we assume that the set of vectors  $\{(a_{1j}, \dots, a_{nj})\,|\, \forall j \in [k]\}$ are linearly independent; then we can perform Gaussian elimination to find a relation between the free variables (exactly $n-k$) and dependent/leading variables. Then, by enumerating over all possible $2^{n-k}$ values of the free variables and finding the corresponding values for leading variables, we can check whether there is a 0/1 solution. This takes~$\poly(n,k) \cdot 2^{n-k}$ time.

In particular, when $k \geq n - O(\log(n))$, $\simulsub$ (with assuming linear independence), it has a polynomial time solution. Whereas, we showed (in \autoref{thm:ssum-to-simulssum}) that given a subset sum instance, we can convert this into a $\simulsub$ instance in polynomial time even with $k = O(\log(n))$.

%{\noindent \textsf{Remark}.}~~

%Talk about SetCoCo ,SETH and parallel..

% Next, we consider another variant of the $\ssum$ problem which allows any non-negative integral combinations (and not only $0-1$ combination). This is called the Unbounded Subset Sum problem ($\ubssum$), which is also known to be $\NP$-complete. 
% \begin{problem}[UBSSUM]\label{prob3:subssum}
% Given $(a_1, \dots, a_n , t) \in \Z_{\geq 0}^{n+1}$, the Unbounded Subset Sum $(\ubssum)$ is to decide whether there exists $(x_1, \dots, x_n) \in \Z_{\geq 0}^n$ such that $\sum_{i \in [n]}\, a_ix_i = t$.
% \end{problem}
%In this work, we give various deterministic algorithms for \autoref{prob1:k-sssum}-\ref{prob2:ham-k}.
%We also show a polynomial time reduction from \autoref{prob2:partition-product} to \autoref{prob1:k-ssssum}. Further, we give an efficient randomized algorithm for \autoref{prob1:k-ssssum} which helps us to achieve a randomized near-linear time algorithm for \autoref{prob2:partition-product}. We also give lowspace algorithm for the former. To the best of our knowledge, our reduction is novel and uses `pseudo'-factorization while the algorithms are algebraic and number theoretic in nature; they mainly build upon the analytic techniques, by Jin and Wu~\cite{jin2018simple} and number-theoretic techniques by Kane~\cite{kane2010unary}.

%--say this in Our results section.The standard  dynamic program takes $O(knt)$ time (and $\Omega(kt)$ space), to output all the solutions, where $k$ is the number of solutions. In this paper, we give a better solution

\subsection{Our results}

In this section, we briefly state our main results. The leitmotif of this paper is to give efficient (time, space) algorithms for all the aforementioned variants of subset sum. 

\subsubsection{Time-efficient algorithms for variants of Subset Sum} \label{sec:efficient-time-algos}
Our first theorem gives an efficient pseudo-linear $\tilde{O}(n+t)$ time {\em deterministic} algorithm for \autoref{prob2:ham-k}, for constant $k$.  
\begin{theorem}[Algorithm for hamming weight]\label{thm1:hamming}
There is a $\tilde{O}(k(n+t))$-time deterministic algorithm for $\hsssum{k}$. 
\end{theorem}

{\noindent $\blacktriangleright$ \textsf{Remark (Optimality)}.}~~We emphasize the fact that \autoref{thm1:hamming} is likely to be {\em near}-optimal for bounded $k$, due to the following argument. An $O(t^{1 - \epsilon})$ time algorithm for $\hsssum{1}$ can be directly used to solve $\ksssum{1}$, as discussed above. By using the {\em randomized} reduction (\autoref{thm:ssum-to-ussum}), this would give us a randomized $n^{O(1)}t^{1-\epsilon}$-time algorithm for $\ssum$. But, in \cite{abboud2019seth} the authors showed that $\ssum$ does not have $n^{O(1)}t^{1-\epsilon}$ time algorithm unless the Strong Exponential Time Hypothesis (SETH) is false.

\paragraph{\autoref{thm1:hamming} is better than the trivial.} Consider the usual `search-to-decision' reduction for subset sum: First try to include $a_1$ in the subset, and if it is feasible then we subtract $t$ by $a_1$ and add $a_1$ into the solution, and then continue with $a_2$, and so on. This procedure finds a single solution, but if we implement it in a recursive way then it can find all the $k$ solutions in $k \cdot n \cdot \text{(time complexity for decision version)}$ time; we can think about an $n$-level binary recursion tree where all the infeasible subtrees are pruned.Since number of solutions is bounded by $k$, choosing a prime $p > n+t+k$ suffices in \cite{jin2018simple}, to make the algorithm deterministic. Thus, the time complexity of the decision version is $\tilde{O}((n+t) \log k)$. Hence, from the above, the search complexity is $\tilde{O}(kn(n+t))$ which is {\em worse} than \autoref{thm1:hamming}.

\begin{theorem}[Time-efficient algorithm for $\pprod$]
\label{thm:time-algos-pprod}
There exists a randomized algorithm that solves $\pprod$ in $\tilde{O}(n+t^{o(1)})$ expected-time.

\end{theorem}

{\noindent \textsf{Remarks}.}~~1. The result in the first part of the above theorem is reminiscent of the $\OO(n+t)$ time randomized algorithms for the subset sum problem~\cite{jin2018simple,bringmann2017near}, although the time complexity in our case is the expected time, and ours is better. 

2. The expected time is because to factor an integer $t$ takes expected $\exp(O(\sqrt{\log(t) \log \log (t)}))$ time~\cite{lenstra1992rigorous}. If one wants to remove expected time analysis (and do the worst case analysis), the same problem can be solved in $\tilde{O}(n^2 + t^{o(1)})$ randomized-time. For details, see the end of \autoref{sec:pf-thm1-details}.

3. While it is true that Bellman's algorithm gives $O(nt)$ time algorithm, the state-space of this algorithm can be  improved to (expected) $nt^{o(1)}$-time for $\pprod$, using a similar dynamic algorithm with a careful analysis. For details, see~\ref{sec:pprod-dynamic}.

\subsubsection{Space-efficient algorithms for variants of Subset Sum}
\begin{theorem}[Algorithms for finding solutions in low space]
\label{thm2:algo-lowspace}
There is a $\poly(knt)$-time and $O(\log (knt))$-space deterministic algorithm which solves $\ksssum{k}$.
\end{theorem}

{\noindent $\blacktriangleright$ \textsf{Remark}.}~~When considering low space algorithms outputting multiple values, the standard assumption is that the output is written onto a one-way tape which {\em does not} count into the space complexity; so an algorithm outputting $k n \log n$ bits (like in the above case) could use much less working memory than $kn\log n$; for a reference see~McKay and Williams~\cite{mckay2018quadratic}. %Thus, \autoref{thm2:algo-lowspace} is {\em optimal}, for $k=O(1)$.

\paragraph{\autoref{thm2:algo-lowspace} is better than the trivial.}~~Let us again compare with the trivial search-to-decision reduction time algorithm, as mentioned in \autoref{sec:efficient-time-algos}. For solving the decision problem in low space, we simply use Kane's $O(\log (nt))$-space $\poly(nt)$-time algorithm \cite{kane2010unary}. As explained (and improved) in~\cite{jin2021fast}, the time complexity is actually $O(n^3t)$ and the extra space usage is $\tilde{O}(n)$ for remembering the recursion stack. Thus the total time complexity is $O(kn^4t)$ and it takes $\tilde{O}(n) + O(\log t)$ space. While \autoref{thm2:algo-lowspace} takes $O(\log (knt))$ space and $\poly(knt)$ time. Although our time complexity is worse \footnote{Thm.~\ref{thm2:algo-lowspace} is {\em not about} time complexity; as long as it is pseudopolynomial time it's ok.}, when $k \le 2^{O((n\log t)^{1-\epsilon})}$, for $\epsilon > 0$, our space complexity is {\em better}.

\begin{theorem}(Algorithm for $\pprod$)
\label{thm:lowspace-pprod}
$\pprod$ can be solved deterministically in $O(\log^2 (nt))$ space and $\poly(nt)$-time.
\end{theorem}

{\noindent $\blacktriangleright$~\textsf{Remark}.} We {\em cannot} directly invoke the theorem in \cite[Section~3.3]{kane2010unary} to conclude, since the reduction from $\pprod$ to $\simulsub$ requires $O(n \log (nt))$ space. Essentially, we use the same identity lemma as \cite{kane2010unary} and carefully use the space; for details see~\autoref{sec:low-space-pprod}.

\subsubsection{Reductions among variants of Subset Sum}

Using a pseudo-prime-factorization decomposition, we show that given a target $t$ in $\pprod$, it suffices to solve $\simulsub$ with at most $\log t$ many instances, where each of the targets are also `small', at most $O(\log \log t)$ bits.
\begin{theorem}[Reducing $\pprod$ to $\simulsub$]
\label{thm:pprod-to-sssum}
There is a deterministic polynomial time reduction from $\pprod$ to $\simulsub$.
\end{theorem}

{\noindent $\blacktriangleright$~\textsf{Remark}.}~~The reduction uses $\tilde{O}(n \log t)$ space as opposed to the following chain of reductions: $\pprod \le_{\P} \ssum \le_{\P} \simulsub$.
The first reduction is a {\em natural} reduction, from an input $(a_1,\hdots, a_n,t)$,  which takes $\log$ both sides and adjust (multiply) a `large' $M$ (it could be $O(n \log t)$ bit~\cite{kovalyov2010generic,pferschy2021approximating}) with $\log a_i$, to reduce this to a $\ssum$ instance with $b_i:=\lfloor M \log a_i \rfloor$. Therefore, the total space required could be as large as $\tilde{O}(n^2 \log t)$. The second reduction follows from~\autoref{thm:ssum-to-simulssum}. Therefore, ours is more space efficient. Motivated thus, we give an efficient randomized algorithm for $\simulsub$.

\begin{figure} \label{fig:reductions}
\begin{center}
\begin{tikzcd}
    &  & \mathsf{HWSSUM} \arrow[dd] &  & \\
    &  &    &  &\\
    \simulsub \arrow[rr] &  & \ssum \arrow[uu, shift right=2] \arrow[dd, shift left=2] \arrow[ll, shift right=2] \arrow[rr, shift left=2] &  & \mathsf{SSSUM} \arrow[ll]            \\
    &  &    &  &\\
    &  & \ubssum \arrow[uu] &  & 
\end{tikzcd}
\caption{Reductions among variants of the Subset Sum problem}
\end{center}
\end{figure}
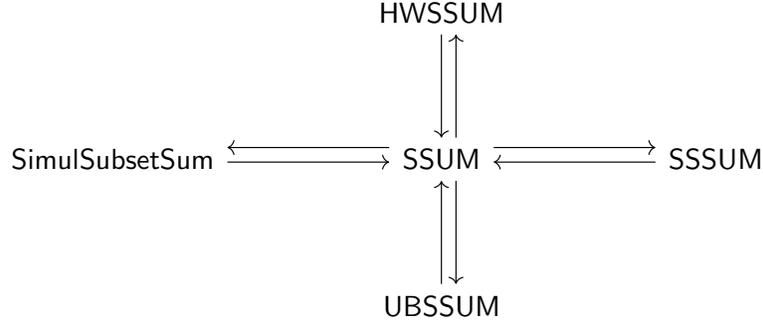

\medskip
In the latter part, we present few reductions among $\ssum$, $\ksssum{k}$ and $\simulsub$ problems. We also extend \autoref{prob1:k-sssum}-\ref{prob2:ham-k} to the {\em unbounded version} of the Subset Sum problem ($\ubssum$) and show similar theorems as above. For details, see~\autoref{sec:pprd-to-ssum}-\ref{sec:unbounded}.

\subsection{Related works}
Before going into the details, we briefly review the state of the art of the problems (\& its variants). After Bellman's $O(nt)$ dynamic solution \cite{bellman57}, Pisinger \cite{pisinger1999linear} first improved it to $O(nt/ \log t)$ on word-RAM models. Recently, Koiliaris and Xu gave a deterministic algorithm \cite{koiliaris2019faster,koiliaris2018subset} in time $\tilde{O}(\sqrt{n}t)$, which is the best deterministic
algorithm so far. Bringmann \cite{bringmann2017near} \& Jin et.al. \cite{jin2018simple} later improved the running time to randomized $\tilde{O}(n + t)$. All these algorithms require $\Omega(t)$ space. Moreover, most of the recent algorithms solve the decision versions. Here we remark that in~\cite{abboud2019seth}, the authors showed that $\ssum$ has no $t^{1-\epsilon} n^{O(1)}$ time algorithm for any $\epsilon > 0$, unless the Strong Exponential
Time Hypothesis (SETH) is false. Therefore, the $\tilde{O}(n + t)$ time bound is likely to be {\em near-optimal}. 

There have been a very few attempts to classically solve $\pprod$ or its variants. It is known to be $\NP$-complete and the reduction follows from the Exact Cover by 3-Sets (X3C) problem \cite[p.~221]{garey1979computers}. Though the knapsack and its approximation versions have been studied~\cite{kovalyov2010generic,pferschy2021approximating}, we do not know many classical algorithms and attempts to solve this, unlike the recent attention for the subset sum problem \cite{bringmann2017near,jin2018simple,jin2021fast,bringmann2021near}. 
% In this paper, we start investigating similar questions in the $\pprod$ regime. 

In \cite{koiliaris2019faster} (also see~\cite[Lemma~2]{koiliaris2018subset}), the authors gave a deterministic $\tilde{O}(nt)$ algorithm that finds all the hamming weights for all realisable targets less than equal to $t$. Their algorithm {\em does not} depend on the number of solutions for a particular target. Compared to this, our~\autoref{thm1:hamming} is {\em faster} when $k=o(n/(\log n)^c)$, for a large constant $c$. Similarly, with the `extra' information of $k$, we give a {\em faster} deterministic algorithm (which even outputs all the hamming weights of the solutions) compared to $\tilde{O}(\sqrt{n}t)$ decision algorithm in \cite{koiliaris2018subset,koiliaris2019faster} (which outputs all the realisable subset sums $\le t$), when $k = o(\sqrt{n}/(\log n)^c)$, for a large constant $c$. Here we remark that the $O(nt)$-time dynamic programming algorithm \cite{bellman57} can be easily modified to find all the solutions, but this gives an $O(n(k + t))$-time (and space) algorithm solution (for more details, see \autoref{appendex:k-SSSUM-trivial}).

On the other hand, there have been quite some work on solving $\ssum$ in $\mathsf{LOGSPACE}$. Elberfeld, Jakoby, and Tantau \cite{elberfeld2010logspace}, and Kane{\cite{kane2010unary}}
(2010) gave $O(\log nt)$ space $\poly(nt)$-time deterministic algorithm, which have been very recently improved to $\tilde{O}(n^2t)$-time and $\poly \log (nt)$ space. On the other hand, Bringmann \cite{bringmann2017near} gave a $nt^{1+\epsilon}$ time, $O(n \log t)$ space {\em randomized} algorithm, which have been improved to $O(\log n \log \log n + \log t)$ space by Jin et.al.\cite{jin2021fast}. Again, most of the algorithms are decision algorithms and do not output the solution set. In contrast to this, our algorithm (\autoref{algo:lowspace}) in \autoref{thm2:algo-lowspace} uses only $O(\log (knt))$ and outputs all the solution sets, which is near-optimal. 

% Finally, we remark that in the proof of \autoref{thm1:hamming}, we extend analytic tools from \cite{jin2018simple} to our advantage (see \autoref{lem:coeff-extraction}), yet our  \autoref{algo:hamming} (for \autoref{thm1:hamming}) is {\em deterministic} (unlike in \cite{jin2018simple}).

Since subset sum can be solved in randomized $\tilde{O}(n+t)$ time~\cite{jin2018simple}, as mentioned before, one obvious way to solve $\pprod$ would be to work with $b_i:=\lfloor M \log a_i \rfloor$ and a $\mathcal{R}$, a range of target values $t'$ which could be as large as $M \log t$ such that $\pprod$ is YES iff subset sum instance with $b_i$ and $t' \in \mathcal{R}$ is YES. But $M$ could be as large as $O(n \cdot \left(\prod_{i} a_i\right)^{1/2})$. Therefore, although there is a randomized near-linear time algorithm for subset sum, when one reduces the instance of $\pprod$ to a subset sum instance, the target becomes very large, failing to give an $\tilde{O}(n+t)$ algorithm. 

% Similarly, \autoref{thm:pprod-to-sssum} along with the reduction in ~\autoref{thm:simul-easier-subsum}, which reduces the $\pprod$ to $\ssum$, actually blows up the target, and fails to give near-linear time algorithm.

Moreover, the general techniques, used for subset sum~\cite{bringmann2017near,jin2018simple,jin2021fast} seem to fail to `directly' give algorithms for $\pprod$. This is exactly why, in this work, the efficient algorithms have been indirect, {\em via} solving $\simulsub$ instances. 
% However, \cite{antonopoulos2021faster} extended the works of \cite{bringmann2017near} and \cite{koiliaris2019faster} to solve $\simulsub$ in randomized time $\tilde{O}(n + \prod_i t_i)$.

% --compare with subsetsum made easy 
% --compare with usual knt vs our algorithm time \& space 
% -- compare Bringman's O(t) ubssum with ours $2^{O(n)}$.
% --Mention the contrast why SOSA is randomized and there you require coeff while here we don't care about zeroness across and thus can be derandomized.

\section{Preliminaries and Notations}

{\bf \noindent Notations.}~$\N, \Z$ and $\Q$ denotes the set of all natural numbers, integers and rational numbers respectively. Let $a,b$ be two $m$-bit integers .Then, $a//b$ denotes $a/b^e$ where $e$ is the largest non-negative integer such that $b^e\,|\,a$. Observe that $a//b$ is not divisible by $b$ and the time to compute $a//b$ is $O(m\log(m)\cdot \log(e))$. 
% Also, $\tilde{O}(N)$ denotes $N \cdot \poly(\log N)$. 

For any positive integer $n > 0$, $[n]$ denotes the set $\{1,2,\dots,n\}$ while $[a,b]$ denotes the set of integers $i$ s.t.~$ a \le i \le b$. Also, $2^{[n]}$ denotes the set of all subsets of $[n]$, while  $\log$ denotes $\log_2$. A weight function $w\,:\,[n]\,\longrightarrow\,[m]$, can be naturally extended to a set $S \in 2^{[n]}$, by defining $w(S)\,:=\,\sum_{i \in S}\,w(i)$. We also denote  $\tilde{O}(g)$ to be $g \cdot \poly(\log g)$.

$\F[x_1, \dots, x_k]$ denotes the ring of $k$-variate polynomials over field $\F$ and $\F[[x_1, \dots, x_k]]$ is the ring of power series in $k$-variables over $\F$. We will use the short-hand notation $\x$ to denote the collection of variables $(x_1, \dots, x_k)$ for some $k$. For any non-negative integer vector $\overline{e} \in \Z^{k}$, $\x^{\overline{e}}$ denotes $\prod_{i=1}^k x_i^{e_i}$. Using these notations, we can will write any polynomial $f(\x) \in \Z[\x]$ as $f(\x) = \sum_{\overline{e} \in S} f_e \cdot  \x^{\overline{e}}$ for some suitable set $S$. 

We denote $\cf_{\x^{\e}}(f)$, as the coefficient of $\x^{\e}$ in the polynomial $f(\x)$ and $\deg_{x_i}(f)$ as the highest degree of $x_i$ in $f(\x)$. Sparsity of a polynomial $f(x_1,\hdots,x_k) \in \F[x_1,\hdots, x_k]$ over a field $\F$, denotes the number of nonzero terms in $f$.

\begin{definition}[Subset Sum problem ($\ssum$)]
Given $(a_1, \dots, a_n, t) \in \Z_{\geq 0}^{n + 1}$, the subset sum problem is to decide whether $t$ is a realisable target with respect to $(a_1, \dots, a_n)$, i.e., there exists $S \subseteq [n]$ such that $\sum_{i \in S}a_i = t$. Here, $n$ is called the {\em size}, $t$ is the {\em target} and any $S \subseteq [n]$ such that $\sum_{i \in S} a_i = t$ is a {\em realisable set} of the subset sum instance.
\end{definition}

{\noindent \bf Assumptions}.~Throughout the paper, we assume that $t \ge \max a_i$ for simplicity. %The latter can be ignored; in that case the complexity parameter will have $n+t$, instead of $t$. % We can work with the case where the input being a multiset can be reduced to the case of a set with little loss in generality and running time (see [\cite{koiliaris2019faster} Section 2.2]), hence for simplicity of
%exposition we assume the input is a set throughout the paper.
%Mahesh-check-this
%Also, we work in the Turing model where basic operations like addition and multiplication over $\F_p$ are not unit-cost unlike Word Ram  model considered in~\cite{jin2018simple}, for simplicity; in the word RAM model our results will give slightly better result shaving one $\log p$ factor. For detailed complexity of basic operations, see~\autoref{appendix:coefficient-extraction}.

\begin{lemma}[{\cite[Isolation Lemma]{mulmuley1987matching}}]
\label{lem:isolation}
Let $n$ and $N$ be positive integers, and let $\mathcal{F}$ be an arbitrary family of subsets of $[n]$. Suppose $w(x)$ is an integer weight given to each element $x \in [n]$ uniformly and independently at random from $[N]$. The weight of $S \in \mathcal{F}$ is defined as $w(S) = \sum_{x \in S}w(x)$. Then, with probability at least $1 - n/N$, there is a unique set $S' \in \mathcal{F}$ that has the minimum weight among all sets of $\mathcal{F}$.
\end{lemma}

%\begin{conjecture}[Artin's conjecture] \label{thm:artin}
%Let $a$ be an integer which is not a perfect square or $-1$. Let $S(a)$ denote the set of prime such that $a$ is a primitive root of $p$. Then, $S(a)$ has a positive asymptotic density inside the set of primes. If $a = a_0 \cdot b^2$ is not a perfect power and $a_0 \not \equiv 1 \bmod\,4$, then $S(a)$ is independent of $a$ and is equal to Artin's constant $C_{Artin}= 0.3739558136192\cdots$.
%\end{conjecture}

\begin{lemma}[Kane's Identity {\cite{kane2010unary}}] \label{lem:kane}
Let $f(x) = \sum_{i=0}^d c_i x^i$ be a polynomial of degree at most $d$ with coefficients $c_i$ being integers. Let $\F_q$ be the finite field of order $q = p^k > d+2$. For $0 \leq t \leq d$, define
\[r_t = \sum_{x \in \F_q^*}x^{q-1-t}f(x) = -c_t \in \F_q\]
Then, $r_t = 0 \iff c_t$ is divisible by $p$.
\end{lemma}

\begin{lemma}[Newton's Identities] \label{lem:newton}
Let $X_1, \hdots, X_n$ be $n \ge 1$ variables. Let $P_m(X_1, \dots, X_n) = \sum_{i=1}^n X_i^m$, be the $m$-th power sum and $E_{m}(X_1, \hdots, X_n)$ be the $m$-th elementary symmetric polynomials, i.e.,~$E_m(x_1,\hdots, x_n)= \sum_{1 \le j_1 \le \hdots \le j_m \le n}\,X_{j_1}\cdots X_{j_m}$, then 
$$m \cdot E_m(X_1, \dots, X_n) \;=\; \sum_{i=1}^m (-1)^{i-1} E_{m-i}(X_1, \dots, X_n) \cdot P_i(X_1, \dots, X_n)\;.$$
\end{lemma}

\begin{remark}
$E_m(X_1,\hdots,X_n)=0$ when $m > n$ (for a quick recollection, see~\href{https://en.wikipedia.org/wiki/Newton\%27s\_identities}{wiki}).
\end{remark}

\begin{lemma}[Vieta's formulas] \label{lem:vieta}
Let $f(x) = \prod_{i=1}^n(x - a_i)$ be a monic polynomial of degree $n$. Then, $f(x) = \sum_{i=0}^n c_ix^i$ where $c_{n-i} = (-1)^{i}E_i(a_1, \dots, a_n), \forall 1 \leq i \leq n$ and $c_n = 1$.
\end{lemma}

%\subsection{Symbolic computational complexity over finite field}

%Polynomial factorization is one of the main problems in algebra and symbolic computations. There have been impressive lot of work over finite field~\cite{berlekamp1970factoring, cantor1981new, evdokimov1994factorization,guo2020deterministic}. In this paper, we are interested in finding the roots (i.e.~the linear factors) of a degree $d$ polynomial $f$ over $\F_p$ of characteristic $p$. We use $\tilde{O}$, since usually the authors use number of operations in $\F_p$ (for eg.~see \cite[Section~1.3]{grenet2016deterministic} for exactness). Also, we remark that square-free factorization takes $\tilde{O}(d \log p)$ time~\cite{berlekamp1970factoring} and \cite[Section~14.6]{von2013modern}.

% \begin{lemma}\label{lem:poly-factoring-unconditional}
% Given a polynomial $f$ of degree $d$ over a finite field $\F_p$, , there exists a deterministic algorithm which finds roots of $f$ over $\F_p$ in time $\tilde{O}(dp^{1/2})$, \cite[Theorem~6]{bourgain2015character},\cite{shoup1990deterministic} and~\cite[Section~1.1]{guo2020deterministic}).
% \end{lemma}

\begin{lemma}[Polynomial division with remainder~{\cite[Theorem 9.6]{von2013modern}}]\label{lem:poly-division}
Given a $d$-degree polynomial $f$ and a linear polynomial $g$ over a finite field $\F_p$, there exists a deterministic algorithm that finds the quotient and remainder of $f$ divided by $g$ in $\tilde{O}(d\log p)$-time.
\end{lemma}

Next, we define ord$_p(a)$; this notion will be important later.
\begin{definition}[Order of a number $\bmod~p$] \label{def:ord}
The order of $a ~(\bmod~p)$, denoted as ord$_p(a)$ is defined to be the {\em smallest} positive integer $m$ such that $a^m \equiv 1~\bmod~p$.
\end{definition} 

%Note that when $p$ is prime, ord$_p(a)$ is clearly finite since $a^{p-1} \equiv 1~\bmod~p$, from Fermat's Little Theorem. Emil Artin (1927, see \cite{moree2012artin}) conjectured that for any non-square $a \in \mathbb{Z}\backslash \{-1\}$, there exist infinitely many primes $p$ such that $a$ is a {\em primitive root} modulo $p$, i.e. ord$_p(a) = p-1$. There has been impressive amount of work done to understand behaviour and distribution of ord$_p(a)$ \cite{gupta1984remark,erdos1999order,chinen2004distribution}. In particular, we have the following.

\begin{theorem}[\cite{shparlinski1996finding}] \label{thm:primitive-root-finding}
There exists a $\tilde{O}(p^{1/4+\epsilon})$ time algorithm to determinstically find a primitive root over $\F_p$.
\end{theorem}

%Here is an important theorem which guarantees existence of prime in a small interval.
\begin{theorem}[\cite{nagura1952interval}] \label{thm:prime-interval}
For $n \ge 25$, there is  a prime in the interval $[n, \frac{6}{5} \cdot n]$.
\end{theorem}

The following is a naive bound, but it is sufficient for our purpose.
\begin{lemma}\label{lem:inequality}
For integers $a \ge b \ge 1$, we have $(a/b)^b \le 2^{2\sqrt{ab}}$.
\end{lemma}
\begin{proof}
Let $x = \sqrt{a/b}$. We need to show that $x^{2b} \le 2^{2bx}$, which is trivially true since $x \le 2^x$, for $x \ge 1$.
\end{proof}

% \begin{definition}[Unbounded $\ssum$ ($\ubssum$)]
% Given $(a_1, \dots, a_n , t) \in \Z_{\geq 0}^{n+1}$, the unbounded $\ssum$ is to decide whether there exists $(x_1, \dots, x_n) \in \Z_{\geq 0}^n$ such that $\sum_{i=1}^n a_ix_i = t$.
% \end{definition}

% Hansen et al.\cite{hansen1996testing} gave an $O(n + a_1^2)$ deterministic algorithm that solves $\ubssum$. There exists a reduction from $\ssum$ to $\ubssum$ and vice versa which implies that $\ubssum$ is also an $\NP$-Complete problem.

% \begin{definition}[Lattice]
% A lattice $\latt(B)$ where $B = [b_1, \dots, b_n] \in \Q^{m \times n}$ is the set of all integer combination of column vectors of $B$, i.e., 
%     \[\latt(B) = \{Bz | \forall z \in \Z^n\}\]
% \end{definition}

% \begin{definition}[Closest Vector Problem ($\cvp_{\ell}$)]
% Given $B \in \Q^{m \times n}, t \in \Q^{m}$, the closest vector problem asks for a closest lattice vector $v \in \latt(B)$ to $t$, i.e.,
% \[\norm{v-t}_{\ell} \leq \norm{w-t}_{\ell}, \forall w \in \latt(B)\]
% \end{definition}

% $\cvp_{\ell}$ is an NP-complete problem for all $\ell$-norm. The fastest known algorithm to solve $\cvp_2$ takes $2^{n + o(n)}$ time and space given by Aggarwal et al. \cite{aggarwal2015solving}.

\section{Hardness results}
In this section, we prove some hardness results. These proofs are very standard, still after different feedback and reviews, we give the details for the brevity. 

Some of the algorithms presented in this paper consider that the number of solutions is {\em bounded} by a parameter $k$. This naturally raises the question whether the $\ssum$\ problem is hard, when the number of solutions is bounded. We will show that this is true even for the case when $k = 1$, i.e., $\usssum$\ is NP-hard under {\em randomized} reduction.

\begin{theorem}[Hardness of $\usssum$]
\label{thm:ssum-to-ussum}
There exists a randomized reduction which takes a $\ssum$\,instance $\mathcal{M} = (a_1, \dots, a_n, t) \in \Z_{\geq 0}^{n+1}$, as an input, and produces multiple $\ssum$\; instances $\mathcal{SS}_\ell = (b_1, \dots, b_n, t^{(\ell)})$, where $\ell \in [2n^2]$, such that if 
\begin{itemize}
 \setlength\itemsep{0.6mm}
    \item $\mathcal{M} \text{ is a YES instance of $\ssum$} \implies \exists \ell \text{ such that $\mathcal{SS}_{\ell}$ is a YES instance of $\usssum$}$;
    \item $\mathcal{M} \text{ is a NO instance of $\ssum$} \implies \forall \ell, \mathcal{SS}_{\ell} \text{ is a NO instance of $\usssum$}$.
\end{itemize}
\end{theorem}

\begin{proof}
The core of the proof is based on the \autoref{lem:isolation} (Isolation lemma). The reduction is as follows. Let $w_1, \dots, w_n$ be chosen {\em uniformly at random} from $[2n]$. We define $b_i = 4n^2a_i + w_i, \forall i \in [n]$ and the $\ell^{th}$ $\ssum$\; instance as $\mathcal{SS}_{\ell} = (b_1, \dots, b_n, t^{(\ell)} = 4n^2 t + \ell)$. Observe that all the new instances are different only in the target values $t^{(\ell)}$.

Suppose $\mathcal{M}$ is a YES instance, i.e., $\exists S \subseteq [n]$ such that $\sum_{i \in S} a_i = t$. Then, for $\ell = \sum_{i \in S} w_i$, the $\mathcal{SS}_{\ell}$ is a YES instance, because

\[
\sum_{i \in S} b_i - t^{(\ell)}\,=\,4n^2\left(\sum_{i \in S} a_i-t\right) \,-\, \left(\ell-\sum_{i \in S} w_i\right)\,=\, 0\;.
\]
If $\mathcal{M}$ is a NO instance, consider any $\ell$ and $S \subseteq [n]$. Since $\mathcal{M}$ is a NO instance, $4n^2(\sum_{i \in S} a_i-t)$ is a non-zero multiple of $4n^2$, whereas $|\ell - \sum_{i \in S}w_i| < 4n^2$, which implies that 

\[
4n^2(\sum_{i \in S} a_i-t) - (\ell - \sum_{i \in S}w_i) \neq 0 \implies \sum_{i \in S} b_i \neq t^{(\ell)}\;.\]

Hence, $\mathcal{SS}_{\ell}$ is also a NO instance.

We now show that if $\mathcal{M}$ is a YES instance, then one of $\mathcal{SS}_{\ell}$ is a $\usssum$. Let $\mathcal{F}$ contain all the solutions to the $\ssum$\; instance $\mathcal{M}$, i.e.,~$\mathcal{F} = \{S | S \subseteq [n], \sum_{i \in S} a_i = t\}$. Since $w_i$'s are chosen uniformly at random, \autoref{lem:isolation} says that there exists a {\em unique} $S \in \mathcal{F}$, such that $w(S) = \sum_{i \in S} w_i$, is {\em minimal} with probability at least $1/2$. Let us denote this minimal value $w(S)$ as $\ell^*$. Then, $\mathcal{SS}_{\ell^*}$ is $\usssum$\; because $S$ is the only subset such that $\sum_{i \in S} w_i = \ell^*$.
\end{proof}

Next, we present a simple deterministic Cook's reduction from $\ssum$ to $\sssum{2}$. It is obvious to see that $\sssum{2} \in \NP$ which implies that $\sssum{2}$ is NP-complete. 

\begin{theorem}[Hardness of $\sssum{2}$] \label{thm:ssum-to-simulssum}
There is a deterministic polynomial time reduction from $\ssum$ to $\sssum{2}$.
\end{theorem}
\begin{proof}
Let $(a_1, \dots, a_n, t)$ be an instance of $\ssum$. Consider the following $\sssum{2}$ instances, $\mathcal{S}_b = [(a_1, \hdots, a_n, t),\, (1, 0, \dots, 0, b)]$, where $b \in \{0,1\}$. If the $\ssum$ instance is NO, then {\em both} the $\sssum{2}$ are also NO. If the $\ssum$ instance is a YES, then we argue that one of the $S_b$ instance must be YES. If $\ssum$ instance has a solution which contains $a_1$, then $\mathcal{S}_1$ is a YES instance whereas if it does not contain $a_1$, then $S_0$ is a YES instance.
\end{proof}

\paragraph{Extension to~$\sssum{\log}$.}
The above reduction can be trivially extended to reduce $\ssum$ to $\simulsub$, with number of $\ssum$ instances $k=O(\log n)$. In that case we will work with instances $S_{\boldsymbol b}$, for ${\boldsymbol b} \in \{0,1\}^k$. Since the number of instances is $2^{k} = \poly(n)$, the reduction goes through.

We will now show that $\sssum{2}$ reduces to $\ssum$ which again can be generalised to $\simulsub$, for any number of $\ssum$ instances $k$. 
% (what happens for $k=3$?????? Is it constant/log)

\begin{theorem}[$\sssum{2}$ is {\em easier} than $\ssum$] \label{thm:simul-easier-subsum}
There is a deterministic polynomial time reduction from $\sssum{2}$ to $\ssum$.
\end{theorem}
\begin{proof}
Let $[(a_1, \dots, a_n, t_1), (b_1, \dots, b_n, t_2)]$ be a $\sssum{2}$ instance where without loss of generality~$t_1 \leq t_2$. Also, we can assume that $t_1 \leq \sum_{i=1}^n a_i$, otherwise it does not have a solution.

Now, consider the $\ssum$ instance $(\gamma b_1 + a_1, \dots, \gamma b_n + a_n, \gamma t_2 + t_1)$, where $\gamma := 1 + \sum_{i=1}^n a_i$. If the $\sssum{2}$ instance is YES, this implies that there exist $S \subseteq [n]$ such that $\sum_{i \in S} a_i = t_1$ and $\sum_{i \in S} b_i = t_2$. This implies that $\sum_{i \in S} \gamma b_i + a_i = \gamma t_2 + t_1$ and hence the $\ssum$ instance is also YES.

Now, assume that the $\ssum$ instance is YES, i.e., there exists $S \subseteq [n]$ such that $\sum_{i \in S} \gamma b_i + a_i = \gamma t_2 + t_1$. This implies that $\gamma ( t_2 - \sum_{i \in S}  b_i ) + (t_1 - \sum_{i \in S} a_i )= 0$. If $t_1 \ne \sum_{i \in S} a_i$, then from the previous equality, $(t_1 - \sum_{i \in S} a_i)$ is a non-zero multiple of $\gamma \implies |\,t_1 - \sum_{i \in S} a_i\, | \ge \gamma$. However, by our assumption, 

\[
t_1 \;\leq\; \sum_{i=1}^n a_i \;\implies\; t_1 - \sum_{i \in S} a_i \;\leq \;\sum_{i \in [n] \setminus S} a_i \;<\; 1+ \sum_{i \in [n]} a_i\;=\;\gamma\;.\]
Moreover, $t_1 - \sum_{i \in S} a_i > -\gamma$, holds trivially, since $\gamma > \sum_{i \in [n]} a_i$ and $t_1 >0$. Therefore, $|\,t_1 - \sum_{i \in S} a_i\, | < \gamma$ which implies that {\em both} $t_1 - \sum_{i \in S} a_i = 0$ and $t_2 - \sum_{i \in S}  b_i = 0$. Hence, the $\sssum{2}$ instance is also YES.
\end{proof}

\section{Time-efficient algorithms}

\subsection{Time-efficient algorithm for \texorpdfstring{$\hsssum{k}$}{}} 
\label{sec:pf-thm1}

In this section, we present an $\tilde{O}(k(n+t))$-time deterministic algorithm for outputting all the hamming weight of the solutions, given a $\hsssum{k}$ instance, i.e.,~there are only at most~$k$-many solutions to the $\ssum$ instance $(a_1,\hdots,a_n,t) \in \Z_{\ge 0}^{n+1}$. The basic idea is simple: We want to create a polynomial whose roots are of the form $\mu^{w_i}$, so that we can first find the roots $\mu^{w_i}$ (over $\F_q$), and from them we can find~$w_i$. To achieve that, we work with $k$-many polynomials $f_j:=\prod_{i=1}^n (1+\mu^{j} \cdot x^{a_i})$, for $j \in [k]$. Note that the coefficient of $x^t$ in $f_j$ is of the form $\sum_{i \le k} \lambda_i \cdot \mu^{jw_i}$ (\autoref{cl:eqs-ham}). By Newton's Identities (\autoref{lem:newton}) and Vieta's formulas (\autoref{lem:vieta}), we can now {\em efficiently} construct a polynomial whose roots are $\mu^{w_i}$. The details are given below.

\begin{proof}[Proof of \autoref{thm1:hamming}]
We start with some notations that we will use throughout the proof.

\paragraph{Basic notations.} Assume that the $\ssum$ instance $(a_1, \dots, a_n, t) \in \Z_{\ge 0}^{n+1}$ has {\em exactly} $m\ (m \le k)$ many solutions, and they have $\ell$ many {\em distinct} hamming weights $w_1, \hdots, w_{\ell}$; since two solutions can have same hamming weight, $\ell \le m$. Moreover, assume that there are $\lambda_i$ many solutions which appear with hamming weight $w_i$, for $i \in [\ell]$. Thus, $\sum_{i \in [\ell]}\lambda_i=m \le k$.

\paragraph{Choosing prime $q$ and a primitive root $\mu$.} We will work with a fixed $q$ in this proof, where $q> n+k+t:=M$ (we will mention why such a requirement later). We can find a prime $q$ in $\tilde{O}(n+k+t)$ time, since we can go over every element in the interval $[M, 6/5 \cdot M]$, in which we know a prime exists (\autoref{thm:prime-interval}) and primality testing is efficient \cite{agrawal2004primes}.  Once we find $q$, we choose $\mu$ such that $\mu$ is a {\em primitive root} over $\F_q$, i.e.,~ord$_q(\mu)=q-1$. This $\mu$ can be found in $\tilde{O}((n+k+t)^{1/4+\epsilon})$ time using \autoref{thm:primitive-root-finding}. Thus, the total time complexity of this step is $\tilde{O}(n+k+t)$.

\paragraph{The polynomials.} Define the $k$-many univariate polynomials as follows:
$$ f_{j}(x)\;:=\;\prod_{i \in [n]}\,(1+\mu^j x^{a_i})\;,\,\forall\,j\,\in\,[k]\;.$$
We remark that we do not know $\ell$ apriori, but we can find $m$ efficiently.
\begin{claim}[Finding the exact number of solutions] \label{cl:finding-numer-of-sol}
Given a $\hsssum{k}$ instance, one can find the exact number of solutions, $m$, deterministically, in $\tilde{O}((n+t))$ time.
\end{claim}
\begin{proof}
Use \cite{jin2018simple} (see~\autoref{lem:coeff-extraction}, for the general statement) which gives a deterministic algorithm to find the coefficient of $x^t$ of $\prod_{i \in [n]}\,(1+x^{a_i})$ over $\F_q$; this takes time~$\tilde{O}((n+t))$. 
\end{proof}
Since we know the exact value of $m$, we will just work with $f_j$ for $j \in [m]$, which suffices for our algorithmic purpose. Here is an important claim about coefficients of $x^t$ in $f_j$'s.

\begin{claim}\label{cl:eqs-ham}
$C_j = \cf_{x^t}(f_j(x))\;=\;\sum_{i \in [\ell]}\,\lambda_i\cdot \mu^{jw_i}$, for each $j \in [m]$.
\end{claim}
\begin{proof}
If $S \subseteq [n]$ is a solution to the instance with hamming weight, say $w$, then this will contribute $\mu^{jw}$ to the  coefficient of $x^t$ of $f_j(x)$. Since, there are $\ell$ many weights $w_1, \hdots, w_{\ell}$ with multiplicity $\lambda_1, \hdots, \lambda_{\ell}$, the claim easily follows.  
\end{proof}
Using \autoref{lem:coeff-extraction}, we can find $C_j~\bmod~q$ for each $j \in [m]$ in $\tilde{O}((n+t \log (\mu j)))$ time, owing total $\tilde{O}(k(n+t))$, since $q=O(n+k+t)$, $\mu \le q-1$, and $\sum_{j \in [m]} \log j = \log (m!) \le \log (k!)=\tilde{O}(k)$.

Using the Newton's Identities (\autoref{lem:newton}), we have the following relations, for $j \in [m]$:
\begin{equation}\label{eq:thm1-relation}
    E_{j}\left(\mu^{w_1},\hdots, \mu^{w_{\ell}}\right) \,\equiv\, j^{-1} \cdot \left(\sum_{i=1}^j (-1)^{i-1}\,E_{j-i}\left(\mu^{w_1}, \dots, \mu^{w_k}\right) \cdot P_i\left(\mu^{w_1},\hdots, \mu^{w_{\ell}}\right)\right)\;\bmod~q\,.
\end{equation}
In the above, by $E_j(\mu^{w_1},\hdots,\mu^{w_{\ell}})$, we mean
$E_j(\underbrace{\mu^{w_1}, \hdots, \mu^{w_1}}_{\lambda_1~\text{times}}, \underbrace{\mu^{w_2}, \hdots, \mu^{w_2}}_{\lambda_2~\text{times}}, \hdots, \underbrace{\mu^{w_{\ell}}, \hdots, \mu^{w_{\ell}}}_{\lambda_{\ell}~\text{times}})$, and similar for $P_j$. Since $q>k$, $j^{-1} \bmod~q$ exists, and thus the above relations are valid. Here is another important and obvious observation, just from the definition of $P_j$'s:
\begin{observation}
\label{obs:thm1-cf-relation}
For $j \in [k]$, $C_j \equiv P_j\left(\mu^{w_1},\hdots, \mu^{w_{\ell}}\right)\bmod~q$.
\end{observation}

Note that we know $E_0 = 1$ and $P_j$'s (and $j^{-1} \bmod q$) are already computed. To compute $E_j$, we need to know $E_1, \hdots, E_{j-1}$ and additionally we need $O(j)$ many additions and multiplications. Suppose, $T(j)$ is the time to compute $E_1, \hdots, E_j$. Then, the trivial complexity is $T(m) \le \tilde{O}(k^2)+\tilde{O}(k(n+t))$. But one can do better than $\tilde{O}(k^2 )$ and make it $\tilde{O}(k )$ (i.e~solve the recurrence, using FFT),~owing the total complexity to $T(m) \le \tilde{O}(k(n+t))$ (since $q=O(n+k+t)$). For details, see \autoref{appendix-recurrence-fft}. 

Once, we have computed $E_j$, for $j \in [m]$, define a new polynomial 

\[
g(x)\;:=\; \sum_{j=0}^{m}\,(-1)^{j} \cdot E_j(\mu^{w_1},\hdots, \mu^{w_{\ell}}) \cdot x^j\;.\]
Using \autoref{lem:vieta}, it is immediate that $g(x) = \prod_{i=1}^{\ell} (x - \mu^{w_i})^{\lambda_i}$. Further, by definition, deg$(g)=m$. From $g$, now we want to extract the roots, namely~$\mu^{w_1}, \hdots, \mu^{w_{\ell}}$ over $\F_q$. We do this, by checking whether $(x-\mu^i)$ divides $g$, for $i \in [n]$ (since $w_i \le n$). Using \autoref{lem:poly-division}, a single division with remainder takes $\tilde{O}(k)$, therefore, the total time to find all the $w_i$ is $\tilde{O}(nk) = \tilde{O}(nk)$.

Here, we {\em remark} that we do not use the determinstic root finding or factoring algorithms (for e.g.~\cite{shoup1990deterministic,bourgain2015character}), since it takes $\tilde{O}(mq^{1/2}) = \tilde{O}(k \cdot (k+t)^{1/2})$ time, which could be larger than $\tilde{O}(k(n+t))$. 
% Root finding step takes additional $\tilde{O}(mq^{1/2})$ time \autoref{lem:poly-factoring-unconditional}.  It remains to find $w_i$ from $\mu^{w_i} \bmod~q$. 

\paragraph{Reason for choosing $q$ and $\mu$.} In the hindsight, there are three important properties of the prime $q$ that will suffice to successfully output the $w_i$'s using the above described steps:
\begin{enumerate}
    \item Since, \autoref{lem:coeff-extraction} {\em requires} to compute the inverses of numbers upto $t$, hence, we would want $q > t$.
    \item While computing $E_{j}(\mu^{w_1}, \dots, \mu^{w_k})$ using \autoref{lem:newton} in the above, one should be able to compute the inverse of all $j$'s less than equal to $m$. So, we want $q > m$,.
    \item To obtain $w_i$ from $\mu^{w_i} \bmod~q$, we want ord$_{q}(\mu)>n$ (for definition see \autoref{def:ord}). Since, $w_i \le n$, this would ensure that we have found the correct $w_i$.
\end{enumerate}

Here, we remark that we do not need to concern ourselves about the `largeness' of the coefficients of $C_j$ and make it nonzero $\bmod~q$, as  required in \cite{jin2018simple}. For the first two points, it suffices to choose $q >k+t$. Since $\mu$ is a primitive root over $\F_q$, this guarantees that ord$_q(\mu) = q-1 > n$ and thus we will find $w_i$ from $\mu^{w_i}$ correctly.

% To find $w_i$, from $\mu^{w_i}$, we just brute force and compute every $\mu^{i} \bmod~q$, for $i \in [n]$ and just check equality with the  $\ell$ many roots. This step takes $\tilde{O}(n \log^2 q)$.

\paragraph{Total time complexity.} The complexity to find the correct $m, q$ and $\mu$ is $\tilde{O}(n+k+t)$. Finding the coefficients of $g$ takes $\tilde{O}(k(n+t))$ and then finding $w_i$ from $g$ takes $\tilde{O}(n k)$ time. Thus, the total complexity remains $\tilde{O}(k(n+t))$. 
\end{proof}

{\noindent $\blacktriangleright$~\textsf{Remark}}.   
\label{remark:randomized-not-helping}
The above algorithm can be extended to find the multiplicities $\lambda_i$'s in $\tilde{O}(k(n+t) + k^{3/2})$ by finding the largest $\lambda_i$, by binary search, such that $(x-\mu^{w_i})^{\lambda_i}$ divides $g(x)$. Finding each $\lambda_i$ takes $\tilde{O}(m  \log (\lambda_i))$ over $\F_q$, for the same $q$ as above, since the polynomial division takes $\tilde{O}(m )$ time and binary search introduces a multiplicative $O(\log (\lambda_i))$ term. Since, $\sum_{i \in [\ell]} \log (\lambda_i) = \log \left( \prod_{i \in [\ell]} \lambda_i\right)$, using AM-GM, $\prod_{i \in [\ell]} \lambda_i \le (m/\ell)^{\ell}$, which is maximized at $\ell= \sqrt{m} \le \sqrt{k}$, implying $\sum_{i \in [\ell]} \log (\lambda_i) \le O(\sqrt{k} \log k)$. Since, $m \le k$, this explains the additive $k^{3/2}$ term in the complexity.
    % but can 
    % . For that, we can just use factoring algorithms~\cite{shoup1990deterministic,bourgain2015character}) which takes additional $\tilde{O}(k \cdot (k+t)^{1/2})$ time. This is subsumed by $\tilde{O}(k \cdot (k+t))$. Note that, factoring gives values of $\lambda_i$'s as well. The total complexity, thus becomes $\tilde{O}(k \cdot (k+t))$.
    
\begin{algorithm}[H]
\caption{Algorithm for $\hsssum{k}$}
\label{algo:hamming}

\vspace{1mm}
\KwIn{A $\ksssum{k}$\ instance $a_1, \dots, a_n, t$}
\KwOut{Hamming weights of all subsets $S \subseteq [n]$ such that $\sum_{i \in S} a_i = t$}

Using \autoref{lem:coeff-extraction}, find the number of solutions $m$ for the $\ksssum{k}$ instance $a_1, \dots, a_n, t$ and terminate if $m= 0$\;
Choose a prime $q$ from the interval $[k+t, 6/5 \cdot (k+t)]$ \;
Find a primitive root $\mu$ over $\F_q$\;
\For{$j \in [m]$}{
    Using $q$ in \autoref{lem:coeff-extraction}, find $C_j = \cf_{x^t}(f_j(x))$ where $f_j(x) = \prod_{i \ = n} (1+\mu^jx^{a_i})$\;
}
Compute $E_0, E_1, \dots, E_{m}$ from $P_1, \dots, P_{m}$ where $P_i \equiv C_i \mod{q}$ using FFT\;
$W = \{\}$\;
\For{$i \in [n]$}{
    \If{$(x-\mu^i) \ |\ g(x)$}{
        $W = W \cup \{i\}$\;
    }
}
% Using a polynomial factoring algorithm, factor $g(x) = \sum_{j=0}^m (-1)^j \cdot E_j$ over $\F_q$ into linear terms (with multiplicities $\lambda_1, \hdots, \lambda_{\ell}$ \;
% From the linear terms $(x-\mu^{w_i})$, extract $w_i$, for $i \in [\ell]$, by computing $\mu^i \bmod~q$, for $i\in [n]$ and checking the equality with the $\mu^{w_i}$\;
\KwRet $W$\;
\end{algorithm}

\section{Time-efficient algorithm for \texorpdfstring{$\pprod$}{}}
\label{sec:rand-algo-pprod}

% In this section, we give an $\tilde{O}(n+t)$-time randomized algorithm for $\pprod$ problem. 

\noindent In this section, we give a randomized $\tilde{O}(n+t^{o(1)})$ expected time algorithm for $\pprod$. Essentially, we factor all the entries in the instance in $\OO(n + t^{o(1)})$ expected time.  Once we have the exponents, it suffices to solve the corresponding~$\simulsub$ instance. Now, we can use the efficient randomized algorithm for $\simulsub$ (\autoref{thm:randomized-algos}) to finally solve $\pprod$. So, first we give an efficient algorithm for $\simulsub$.

\begin{theorem}[Algorithm for $\simulsub$]
\label{thm:randomized-algos}
There is a randomized $\tilde{O}(kn+ \prod_{i \in [k]} (2t_i + 1))$-time algorithm that solves $\simulsub$, with target instances $t_1, \hdots, t_k$.
\end{theorem}

\begin{proof}
Let us assume that the input to the $\simulsub$ problem are $k$ $\ssum$ instance of the form $(a_{1j}, \hdots, a_{nj}, t_j)$, for $j \in [k]$. Define a $k$-variate polynomial $f(\x)$, where $\x= (x_1, \dots, x_k)$, as follows:

\[
f(\x) \;=\; \prod_{i=1}^{n}\, \left( 1 \,+\, \prod_{j=1}^{k} \,x_j^{a_{ij}}\right)\;.\]
Here is an immediate but important claim. We denote the monomial $\m := \prod_{i=1}^{k} x_i^{t_i}$ and $\cf_{\m}(f)$ as the coefficient of $\m$ in the polynomial $f(\x)$.
\begin{claim}
There is a solution to the $\simulsub$ instance, i.e., $\exists S \subseteq [n]$ such that $\sum_{i \in S} a_{ij} = t_j, \forall j \in [k]$ {\em iff} $\cf_{\m} (f(\x)) \neq 0$.
\end{claim}
Therefore, it is enough to compute the coefficient of $f(\x)$. The rest of the proof focuses on computing $f(\x)$ efficiently, to find $\cf_{\m}(f)$. 

\smallskip
%Wlog, we can assume that $t_1 \geq t_i, \forall i \in [k]$, otherwise we  can rearrange. 
Let $p$ be prime such that $p \in [N+1, (n+N)^3]$, where $N:= \prod_{i=1}^k (2t_i + 1)$. Define an ideal $\calI$, over $\Z[\x]$ as follows: $\calI := \langle x_1^{t_1+1}, \hdots, x_k^{t_k+1}, p \rangle$. Since, we are interested in $\cf_{\m}(f)$, it suffices to compute  $f(\x)\bmod~\langle x_1^{t_1+1}, \hdots, x_k^{t_k+1}\rangle$, and we do it over a field $\F_p$ (which introduces error); for details, see the proof in the end (Randomness and error probability paragraph).

%We want to compute $f(\x)~\bmod~\calI$. Equivalently, we are interested in computing $f(\x) \bmod~\langle x_1^{t_1+1}, \hdots, x_k^{t^{k}+1} \rangle$, over $\F_p$.
\smallskip
Using \autoref{lem:coefficient-extraction}, we can compute all the coefficient of $\ln(f(\x))\,\bmod\,\calI$ in time $ \tilde{O}(kn+\prod_{i=1}^k t_i)$. It is easy to see that the following equalities hold.

\[
f(\x)~\bmod~\calI \;\equiv\; \exp\left(\ln(f(\x))\right)~\bmod~\calI \;\equiv\;  \exp\left(\ln(f(\x))\bmod\,\calI\right)~\bmod~\calI\;.
\]
Since, we have already computed $\ln(f(\x))\,\bmod\,\calI$, the above equation implies that it is enough to compute the exponential which can be done using~\autoref{lem:multivariate-exponentiation}. This also takes time $\tilde{O}(kn+\prod_{i=1}^k (2t_i + 1))$. 

\paragraph{Randomness and error probability.} Note that there are $\Omega(n+N)^2$ primes in the interval $[N+1, (n+N)^3]$. Moreover, since $\cf_{\m}(f) \le 2^n$, at most $n$ prime factors can divide $\cf_{\m} (f(\x))$. Therefore, we can pick a prime $p$ randomly from this interval in $\poly(\log(n+N))$ time and the the probability of $p$ dividing the coefficient is $O(n+N)^{-1}$. In other words, the probability that the algorithm fails is bounded by $O\left((n+N)^{-1}\right)$. This concludes the proof. \qed

\end{proof}

We now compare the above result with some obvious attempts to solve $\simulsub$, before moving into solving $\pprod$.

{\noindent \bf A detailed comparison with time complexity of \cite{kane2010unary}.}~~Kane~\cite[Section~3.3]{kane2010unary} showed that the above problem can be solved deterministically in $C^{O(k)}$ time and $O(k \log C)$ space, where $C:= \sum_{i,j} a_{ij} + \sum_{j} t_j + 1$, which could be as large as $(n+1) \cdot (\sum_{j \in [k]} t_j) + 1$, since $a_{ij}$ can be as large as $t_j$. As argued in~\cite[Corollary~3.4 and Remark~3.5]{jin2021fast}, the constant in the exponent, inside the order notation, can be as large as $3$ (in fact directly using \cite{kane2010unary} gives larger than $3$; but modified algorithm as used in~\cite{jin2021fast} gives $3$). Use  AM-GM inequality to get
\[
\left( (n+1) \cdot (\sum_{j} t_j )+1\right)^{3k}\,>\, \left( \frac{2}{k} \cdot \sum_{j} t_j + 1\right)^{3k}\;\stackrel{\text{AM-GM}}{\ge}\; \prod_{j=1}^k\,\left(2t_j+1\right)^{3}\;.
\]
Assuming $N =\prod_{j=1}^k\,(2t_j+1)$, our algorithm is near-linear in $N$ while Kane's algorithm~\cite{kane2010unary} takes at $O(N^3)$ time; thus ours is almost a cubic improvement.

\medskip
{\noindent \bf Comparison with the trivial algorithm.}~~It is easy to see that a trivial $O(n\cdot (t_1+1) (t_2+1) \hdots (t_k+1))$ time {\em deterministic} algorithm for $\simulsub$ exists. Since, $t_i \ge 1$, we have 

\[\frac{n}{2} \cdot \prod_{i \in [k]} (1+t_i) \,\ge\, \frac{n}{2} \cdot 2^k \,\ge\, kn\,,\;\;\text{and}\;\; \frac{n}{2} \cdot \prod (1+t_i) \,\ge\, \frac{n}{2^{k+1}} \cdot \prod (2t_i+1)\,.
\]
Here, we used $2(1+x) > (2x+1)$, for any $x \ge 1$. Therefore, $n \cdot \prod_{i \in [k]} (1+t_i) \ge kn + n/2^{k+1} \cdot \prod (2t_i+1)$. Thus, when $k = o(\log n)$, our complexity is better.

 %The deterministic algorithm is based on the proof idea used in~\autoref{lem:kane}. 

\subsection{Proof of \autoref{thm:time-algos-pprod}} \label{sec:pf-thm1-details}

Once we have designed the algorithm for $\simulsub$, we design time-efficient algorithm for \autoref{thm:time-algos-pprod}.

\begin{proof}
Let $(a_1, \dots, a_n, t) \in \Z_{\ge 0}^{n+1}$ be the input for $\pprod$ problem. Without loss of generality, we can assume that all the $a_i$ divides $t$ because if some $a_i$ does not divide $t$, it will never be a part of any solution and we can discard it. Let us first consider the prime factorization of $t$ and $a_j$, for all $j \in [n]$. We will discuss about its time complexity in the next paragraph. Let

    \[t \;=\; \prod_{j=1}^{k} \,p_j^{t_j}\,,~~~~a_i \;=\; \prod_{j=1}^{k} \,p_j^{e_{ij}},\,\forall\,i\,\in\,[n]\,,\]
where $p_j$ are distinct primes and $t_j$ are positive integers and $e_{ij} \in \Z_{\ge 0}$. Since, $p_i \ge 2$, trivially, $\sum_{i=1}^k t_i \le \log(t)$, and $\sum_{i=1}^k e_{ij} \le  \log(t), j \in [n]$.  Also, the number of distinct prime factors of $t$ is at most $O(\log(t)/\log\log(t))$; therefore, $k = O(\log(t)/\log\log(t))$. 

\paragraph{Time complexity of factoring} 
To find all the primes that divides $t$, we will use the factoring algorithm given by Lenstra and Pomerance~\cite{lenstra1992rigorous} which takes expected $t^{o(1)}$ \footnote{Expected time complexity is $\exp(O(\sqrt{\log t \log \log t}))$, which is smaller than $t^{O(1/\sqrt{\log \log t})} =t^{o(1)}$, which will be the time taken in the next step. Moreover, we are interested in {\em randomized} algorithms, hence expected run-time is} time to completely factor $t$ into prime factors $p_j$ (including the exponents $t_j$). Using the primes $p_j$ and the fact that $0 \leq e_{ij} \leq \log(t)$, computing $e_{ij}$ takes $\log^2(t)\log\log(t)$ time, by performing binary search to find the largest $x$ such that $p_j^{x}\,|\, a_i$. So, the time to compute all exponents $e_{i,j}, \forall i \in [n], j \in [k]$ is $O(nk\log^2(t)\log\log(t))$. Since, $k \le O(\log t/\log\log(t))$, the total time complexity is $\tilde{O}(n + t^{o(1)})$.
% By trivially searching for all primes divides $t$ in the interval $[2,\sqrt{t}]$, we can find all the primes $p_j$ (including the exponents $t_j$) in time $\tilde{O}(\sqrt{t})$. Using the primes $p_j$ and the fact that $0 \leq e_{ij} \leq \log(t)$, computing $e_{ij}$ takes $
% \log^2(t)\log\log(t)$ time, by performing binary search to find the largest $x$ such that $p_j^{x}\,|\, a_i$. So, the time to compute all exponents $e_{i,j}, \forall i \in [n], j \in [k]$ is $O(nk\log^2(t)\log\log(t))$. Since, $k \le \log t$, the total time complexity is $\tilde{O}(n + t)$. 

%%--Explain.
\paragraph{Setting up $\simulsub$} Now suppose that $S \subseteq [n]$ is a solution to the $\pprod$ problem, i.e., $\prod_{i \in S} a_i = t$. This implies that 

    \[\sum_{i \in S} \,e_{ij} \;=\; t_j\,,~~\forall\,j \,\in\, [k]\;.\]
In other words, we have a $\simulsub$ instance where the $j^{th}$ $\ssum$ instance is $(e_{1j}, e_{2j}, \dots, e_{nj}, t_j)$, for $j \in [k]$. The converse is also trivially true. We now show that there exists an $\tilde{O}(kn + \prod_{i \in [k]} (2t_i + 1))$ time algorithm to solve $\simulsub$.

\paragraph{Randomized algorithm for $\pprod$} Using \autoref{thm:randomized-algos}, we can decide the $\simulsub$ problem with targets $t_1, \dots, t_k$ in $\tilde{O}(kn + \prod_{i \in [k]}(2t_i + 1))$ time (randomized) while working over $\F_p$ for some suitable $p$ (we point out towards the end). Since $k \leq O(\log(t)/\log\log(t))$, we need to bound the term $\prod_{i \in [k]} (2t_i + 1)$. Note that,
\begin{align*}
    \prod_{i \in [k]}\,(2t_i + 1) &\;=\; \sum_{S \subseteq [k]}\, 2^{|S|}\cdot\left(\prod_{i \in S}\,t_i\right) \\
    &\;\leq\; 2^{2k} \cdot \left(\prod_{i \in [k]} t_i\right)\;.
\end{align*}

We now focus on bounding the term $\prod_{i \in [k]} t_i$. By AM-GM,
\begin{align*}
    \prod_{i \in [k]} t_i \leq \left ( \dfrac{\sum_{i \in [k]} t_i}{k}\right)^{k} 
    &\leq \left ( \dfrac{\log(t)}{k}\right)^{k} \\
    &\leq 2^{O\left(\sqrt{k\log(t)}\right)}~~~~[Lemma~\ref{lem:inequality}] \\
    &\leq 2^{O\left(\sqrt{\log(t)^2/\log\log(t)}\right)} \\
    &\leq t^{O(1/\sqrt{\log\log(t)})} = t^{o(1)}
\end{align*}

%Assuming $\log p = O(\log(n+t))$, the time to solve the above $\simulsub$ is $\tilde{O}(kn + \prod_{i=1}^{k} (2t_i+1)) = \tilde{O}(n + t)$, because $k \leq \log(t)$ and

%     \[\prod_{i=1}^{k}\, (2t_i+1) \;\leq\; c\cdot\prod_{i=1}^{k}\, p_{i}^{t_i} \;=\; O(t)\;,\]
% where $c$ is a constant. The above inequality holds due to the following facts -- 1) if $s$ is any non-negative integer and for any prime $p > 3$, we have $(2s+1) \leq p^s$, 2) whereas we have $(2s+1) \leq 3\cdot2^s$ and $ (2s+1) \leq 2\cdot3^s$. Therefore, the constant $c=6$ works in the above inequality. 

Note that the prime $p$ in the \autoref{thm:randomized-algos} was $p \in [N+1, (n+N)^3]$, where $N:= \prod_{i=1}^k (2t_i+1)-1$. As shown above, we can bound $N=t^{o(1)}$. Thus, $p \le O((n+t^{o(1)})^3)$, as desired. Therefore, the total time complexity is $\tilde{O}(n\log(t)/\log\log(t) + t^{o(1)}) = \tilde{O}(n + t^{o(1)})$. This finishes the proof. \qed
\end{proof}

\paragraph{Removing the expected-time} If one wants to understand the worst-case analysis, we can use the polynomial time reduction from $\pprod$ to $\simulsub$ in \autoref{sec:pprd-to-ssum}. Of course, we will not get prime factorization; but the pseudo-prime factors will also be good enough to set up the $\simulsub$ with similar parameters as above, and the $\simulsub$ instance can be similarly solved in $\tilde{O}(n+t^{o(1})$ time. Since the reduction takes $n^2\poly(\log t)$ time, the total time complexity becomes $\tilde{O}(n^2+t^{o(1)})$.

\section{Space-efficient algorithms}

\subsection{Space-efficient algorithm for \texorpdfstring{$\ksssum{k}$}{}}
\label{sec:k-SSSUM-algos}

In this section, we will present a low space algorithm (\autoref{algo:lowspace}) for finding all the realisable sets for $\ksssum{k}$. Unfortunately, proof of \autoref{thm1:hamming} {\em fails}~to give a low space algorithm, since \autoref{lem:coeff-extraction} requires $\Omega(t)$ space (eventually it needs to store all the coefficients~$\bmod~x^{t+1}$). Instead, we work with a multivariate polynomial $f(x, y_1, \dots, y_n) = \prod_{i=1}^n (1+y_ix^{a_i})$ over $\F_q$,  for a large prime $q = O(nt)$ and its multiple evaluations $f(\alpha,c_1,\hdots,c_n)$, where $(\alpha,c_1,\hdots,c_n)\in \F_q^{n+1}$. 

Observe that, the coefficient of $x^t$ in $f$ is a multivariate polynomial $p_t(y_1, \dots, y_n)$; each of its monomial carries the {\em necessary information} of a solution, for the instance $(a_1, \dots, a_n, t)$. More precisely, $S$ is a realisable set of $(a_1, \dots, a_n, t) \iff \prod_{i \in S} y_i$ is a monomial in $p_t$. And, the sparsity (number of monomials) of $p_t$ is at most $k$. Therefore, it boils down to reconstruct the multivariate polynomial $p_t$ efficiently. We cannot use the trivial multiplication since it takes $\tilde{O}(2^nt)$ time! Instead, we use ideas from~\cite{kane2010unary} and \cite{klivans2001randomness}.

%  Our low space algorithms build upon a fundamental number-theoretic identity \cite{kane2010unary}, and efficient sparse multivariate polynomial reconstruction \cite{klivans2001randomness}. 
\begin{proof}[Proof of \autoref{thm2:algo-lowspace}]
Here are some notations that we will follow throughout the proof.

\paragraph{Basic notations.} Let us assume that there are exactly $m\ (m \le k)$ many realisable sets $S_1, \hdots, S_m$, each $S_i \subseteq [n]$. We remark that for our algorithm we do not need to apriori calculate $m$.

\paragraph{The multivariate polynomial.} For our purpose, we will be working with the following $(n+1)$-variate polynomial: 
$$
f(x,y_1,\hdots,y_n)\;:=\; \prod_{i \in [n]}\,\left(1+ y_i x^{a_i}\right)\;.
$$
Since, we have a $\ksssum{k}$ instance $(a_1,\hdots, a_n,t)$, $\cf_{x^t}(f)$ has the following properties.

\vspace{-.5mm}
\begin{enumerate}
    \item It is an $n$-variate polynomial $p_t(y_1, \dots, y_n)$ with sparsity {\em exactly} $m$.
    \item $p_t$ is a multilinear polynomial in $y_1,\dots, y_n$, i.e.,~individual degree of $y_i$ is at most $1$.
    \item The total degree of $p_t$ is at most $n$.
    \item if $S \subseteq [n]$ is a realisable set, then ${\mathbf y}_S:= \prod_{i \in S} y_i$, is a monomial in $p_t$.
\end{enumerate}

In particular, the following is an immediate but important observation.
\begin{observation}
$p_t(y_1,\hdots,y_n)\,=\,\sum_{i \in [m]}\,\y_{S_i}$\,.
\end{observation}
Therefore, it suffices to know the polynomial $p_t$. However, we cannot treat $y_i$ as new variables and try to find the coefficient of $x^t$ since the trivial multiplication algorithm (involving $n+1$ variables) takes $\exp(n)$-time. This is because, $f(x, y_1, \hdots, y_n)~\bmod~x^{t+1}$ can have $2^n \cdot t$ many monomials as coefficient of $x^i$, for any $i \le t$ can have $2^n$ many multilinear monomials. 

However, if we substitute $y_i=c_i \in \F_q$, for some prime $q$, we claim that we can figure out the value $p_t(c_1,\hdots, c_n)$ from the coefficient of $x^t$ in $f(x,c_1,\hdots,c_n)$  efficiently (see \autoref{cl:eff-eval}). Once we have figured out, we can simply interpolate using the following theorem to reconstruct the polynomial $p_t$. Before going into the technical details, we state the sparse interpolation theorem below; for simplicity we consider multilinearity (though \cite{klivans2001randomness} holds for general polynomials as well). 

\begin{theorem}[\cite{klivans2001randomness}] \label{thm:sparse-interpolation}
Given a black box access to a multilinear polynomial $g(x_1, \dots, x_n)$ of degree $d$ and sparsity at most $s$ over a finite field $\F$ with $|\F| \ge  (nd)^6$, there is a $\poly(snd)$-time and $O(\log (snd))$-space algorithm that outputs all the monomials of $g$.
\end{theorem}
{\noindent $\blacktriangleright$~\textsf{Remark}.}~~
We represent one monomial in terms of indices (to make it consistent with the notion of realisable set), i.e.,~for a monomial $x_1x_5x_9$, the corresponding indices set is $\{1,5,9\}$. Also, we do not include the indices in the space complexity, as mentioned earlier. 

\paragraph{Brief analysis on the space complexity of \cite{klivans2001randomness}.} Klivans and Spielman  \cite{klivans2001randomness}, did not explicitly mention the space complexity. However, it is not hard to show that the required space is indeed $O(\log (snd))$. \cite{klivans2001randomness} shows that substituting $x_i = y^{k^{i-1} \bmod p}$, for some  $k \in [2s^2n]$ and $p > 2s^2n$, makes the exponents of the new univaraite polynomial (in $y$) {\em distinct} (see~\cite[Lemma~3]{klivans2001randomness}); the algorithm actually tries for all $k$ and find the correct $k$. Note that the degree becomes $O(s^2nd)$. Then, it tries to first find out the coefficients by simple univariate interpolation \cite[Section~6.3]{klivans2001randomness} . Since we have blackbox access to $g(a_1,\hdots,a_n)$, finding out a single coefficient, by univariate interpolation (which basically sets up linear equations and solve) takes $O(\log (snd))$ space and $\poly(snd)$ time only. In the last step, to find one coefficient, we can use the standard univariate interpolation algorithm which uses the Vandermonde matrices and one entry of the inverse of the Vandermonde is $\log$-space computable \footnote{In fact Vandermonde determinant and inverse computations are in $\TC^0 \subset$~$\mathsf{LOGSPACE}$, see~\cite{maciel1998threshold}.}. 

At this stage, we know the coefficients (one by one), but we do not know which monomials the coefficients belong. However, it suffices to substitute $x_i=2y^{k^{i-1} \bmod~p}$. Using this, we can find the the correct value of the first exponent in the monomial. For eg.~ if after the correct substitution, $y^{10}$ appears with coefficient say~$5$, next step, when we change just $x_1$, if it does not affect the coefficient $5$, $y_1$ is not there in the monomial corresponding to the monomial which has coefficient $5$, otherwise it is there (here we also use that it is multilinear and hence the change in the coefficient must be reflected). This step again requires univariate interpolation, and  one has to repeat this experiment wrt each variable to know the monomial exactly corresponding to the  coefficient we are working with. We can reuse the space for interpolation and after one round of checking with every variable, it outputs one exponent at this stage. This requires $O(\log (snd)$-space and $\poly(snd)$ time.

With a more careful analysis, one can further improve the field requirement to $|\F| \ge (nd)^6$ only (and not dependent on $s$); for details see \cite[Thm.~5 \& 11]{klivans2001randomness}. %This will remove the $s$ dependence on the $\log (snd)$ term as well, since this term  basically corresponds to $\log p$.

\bigskip
Now we come back to our subset sum problem. Since we want to reconstruct an $n$-variate $m$ sparse polynomial $p_t$ which has degree at most $n$, it suffices to work with $|\F| \ge n^{12}$. However, we also want to use Kane's identity (\autoref{lem:kane}), which requires $q > \deg(f(x,c_1,\hdots,c_n))+2$, and $\deg(f(x,c_1,\hdots,c_n)) \le nt$. Denote $M:=\max(nt+3,n^{12})$. Thus, it suffices to we work with $\F= \F_q$ where $q \in [M,(6/5) \cdot M]$, such prime exists (\autoref{thm:prime-interval}) and easy to find deterministically in $\poly(nt)$ time and $O(\log (nt))$ space using \cite{agrawal2004primes}. In particular, we will substitute $y_i=c_i \in [0,q-1]$.

\begin{claim} \label{cl:eff-eval}
Fix $c_i \in [0,q-1]$, where $q \in [M,(6/5)\cdot M]$. Then, there is a $\poly(nt)$-time and $O(\log (nt))$ space algorithm which computes $p_t(c_1,\hdots,c_n)$ over $\F_q$.
\end{claim}

\begin{proof}
Note that, we can evaluate each $1+c_ix^{a_i}$, at some $x=\alpha \in \F_q$, in $\tilde{O}(\log nt)$ time and $O(\log (nt))$ space. Multiplying $n$ of them takes $\tilde{O}(n \log (nt))$-time and $O(\log (nt))$ space. 

Once we have computed $f(\alpha,c_1,\hdots,c_n)$ over $\F_q$, using Kane's identity (\autoref{lem:kane}), we can compute $p_t(c_1,\hdots, c_n)$, since
$$
p_t(c_1,\hdots,c_n)\;=\;-\,\sum_{\alpha \in \F_q^*}\, \alpha^{q-1-t}f(\alpha,c_1,\hdots,c_n)\;.
$$
As each evaluation $f(\alpha,c_1,\hdots,c_n)$ takes $\tilde{O}(n \log (nt))$ time, and we need $q-1$ many additions, multiplications and modular exponentiations, total time to compute is $\poly(nt)$. The required space still remains $O(\log (nt))$. 
\end{proof}

Once, we have calculated $p_t(c_1,\hdots,c_n)$ efficiently, now we try different values of $(c_1,\hdots,c_n)$ to reconstruct $p_t$ using \autoref{thm:sparse-interpolation}. Since, $p_t$ is a $n$-variate at most $k$ sparse polynomial with degree at most $n$, it still takes $\poly(knt)$ time and  $O(\log (knt))$ space. This finishes the proof. 
\end{proof}

\begin{algorithm}
\caption{Algorithm for $\ksssum{k}$}
\label{algo:lowspace}
\vspace{1mm}
\KwIn{A $\ksssum{k}$\ instance $a_1, \dots, a_n, t$}
\KwOut{All realisable subsets $S \subseteq [n]$ such that $\sum_{i \in S} a_i = t$}

Pick a prime $q \in \{M,(6/5) \cdot M\}$ where $M = max(nt + 3, n^{12})$\;
Let $\mathcal{O}$ be the algorithm mentioned in \autoref{thm:sparse-interpolation}\; 
\For{each $p_t(c_1, \dots, c_n)$ query requested by $\mathcal{O}$}{
    Send $-(\underset{\alpha \in \F_q^*}{\sum} \alpha^{q-1-t} f(\alpha, c_1, \dots, c_n))$ to $\mathcal{O}$\;
}
$p_t$ be the polynomial return by $\mathcal{O}$\;
$\mathcal{F} = \{\}$\;
\For{each monomial $\y_{S}$ in $p_t$}{
    $\mathcal{F} = \mathcal{F} \cup \{S\}$
}
\KwRet $\mathcal{F}$\;

\end{algorithm}

\subsection{Space-efficient algorithm for \texorpdfstring{$\pprod$}{}}
\label{sec:low-space-pprod}

The proof of~\autoref{thm:lowspace-pprod} uses the idea of reducing $\pprod$ to $\simulsub$ and then solving $\simulsub$ by computing the coefficient of $f(\x)= \prod_{i=1}^{n}\, \left( 1 \,+\, \prod_{j=1}^{k} \,x_j^{a_{ij}}\right)$ where $\x=(x_1,\hdots, x_k)$ using an extension of \autoref{lem:kane}. We cannot directly use~ \cite{kane2010unary} as it requires large space ($O(n\log(nt))$ space to be precise) to store the $\simulsub$ instance. Instead we compute the coefficient of $f(\x)$ without storing the $\simulsub$ instance using~\autoref{lem:kane-extension}.

The low space algorithm presented in this proof depends on the generalisation of \autoref{lem:kane}. Here we present Kane's identity for {\em bi}variate polynomials which can be easily extended to $k$-variate polynomials.
\begin{lemma}[Identity lemma {\cite{kane2010unary}}]
\label{lem:kane-extension}
Let $f(x,y) = \sum_{i=0}^{d_1} \sum_{j = 0}^{d_2} c_{i,j}x^iy^j$ be a polynomial of degree at most $d_1+d_2$ with coefficients $c_{i,j}$ being integers. Let $\F_q$ be the finite field of order $q = p^m > \max(d_1,d_2) + 1$. For $0 \leq t_1 \leq d_1, 0 \leq t_2 \leq d_2$, define
\[
r_{t_1, t_2} = \sum_{x \in \F_q^*} \sum_{y \in \F_q^*} x^{q-1-t_1}y^{q-1-t_2} f(x,y) = c_{t_1, t_2} \in \F_q
\]
\end{lemma}
\begin{proof}
Let $n$ be a positive integer, then the two following identities hold:
\begin{enumerate}
    \item[Identity 1.] $\sum_{x \in \F_q^*} x^n = -1$ if $q-1\, |\, n$ because $x^n = x^{(q-1)m} = 1$ due to Fermat's Little theorem. 
    \item[Identity 2.] $\sum_{x \in \F_q^*} x^n = 0$, if $q-1 \nmid\,n$. This is because we can rewrite the summation as $\sum_{i=0}^{q-2} g^{i\cdot n} = \dfrac{g^{n(q-1)}-1}{g^{n}-1} = 0$ where $g$ is a generator of $\F_q^*$.
\end{enumerate}
Let us now consider $\sum_{x \in \F_q^*} \sum_{y \in \F_q^*} x^{q-1-t_1}y^{q-1-t_2} f(x,y)$.
\begin{align*}
    \sum_{x \in \F_q^*} \sum_{y \in \F_q^*} x^{q-1-t_1}y^{q-1-t_2} f(x,y) &= \sum_{x \in \F_q^*} \sum_{y \in \F_q^*} x^{q-1-t_1}y^{q-1-t_2} \left( \sum_{i=0}^{d_1} \sum_{j = 0}^{d_2} c_{i,j}x^iy^j \right)\\
    &= \sum_{i=0}^{d_1} \sum_{j = 0}^{d_2} c_{i,j}\left( \sum_{x \in \F_q^*} \sum_{y \in \F_q^*} x^{q-1-t_1+i}y^{q-1-t_2+j}  \right) \\
    &= \sum\limits_{\substack{i \in [0,d_1]\setminus \{t_1\} \\ j \in [0,d_2] \setminus \{t_2\}}} c_{i,j}\left( \sum_{x \in \F_q^*} \sum_{y \in \F_q^*} x^{q-1-t_1+i}y^{q-1-t_2+j}  \right) \\
    &\;\;\;\;+ c_{t_1, t_2} \\
    &= c_{t_1, t_2}
\end{align*} 
Observe that when $i \in [0,d_1] \setminus\{t_1\}$, we have $\sum_{x \in \F_q^*} x^{q-1 - t_1 + i} = 0$ because $|i-t_1| \leq d_1 < q-1 \implies q-1+i-t_1$ is not a multiple of $q-1$. The same goes for $j \in [0,d_2]\setminus \{t_2\}$.

\end{proof}

% We will now show how to use \autoref{thm:kane-extension} to solve $\sssum{2}$. Let $[(a_1, \dots, a_n, t_1)$ , $ (b_1, \dots, b_n, t_2)]$ be an instance of $\sssum{2}$. We define a bi-variate polynomial $f(x,y) = \prod_{i=1}^n (1+x^{a_i}y^{b_i})$. It is easy to see that the $\sssum{2}$ instance is YES $\iff$ the coefficient of $x^{t_1}y^{t_2}$ is non-zero in $f(x,y)$. The coefficient of $x^{t_{1}}y^{t_2}$ can be computed using the \autoref{thm:kane-extension}. The degree of $f(x,y)$ is at most $\sum_{i=1}^n a_i + \sum_{i=1}^n b_i \leq n (t_1 + t_2)$. Therefore, the prime $p$ required in \autoref{thm:kane-extension} must be larger than $n(t_1 + t_2) + 2$. Using \autoref{thm:prime-interval}, we can find the required prime by going through every element in the interval $[n (t_1 + t_2) + 3, 6/5 \cdot (n (t_1 + t_2) + 3)]$ and checking whether it is a prime using \cite{agrawal2004primes}. This requires $\poly(n,t_1,t_2)$ time but only log-space. Furthermore, the time complexity of computing the summation in \autoref{thm:kane-extension} is again $\poly(n,t_1,t_2)$.

{\noindent $\blacktriangleright$~\textsf{Remark}.}~~\autoref{lem:kane-extension} can be easily extended to $k$ variables which was used by the authors of \cite{kane2010unary} to solve $\simulsub$ with $k$ many $\ssum$ instances in space $O(k\log(n \sum_{i=1}^k t_i))$ and time $(poly(n,t_1, \dots, t_k))^{O(k)}$. In this case, the order of the finite field must be greater than $\max(d_1,\hdots, d_k) + 1$ where $d_i$'s are the individual degrees of the polynomial. 

% \end{proof}
\paragraph{Issue with directly invoking~\cite{kane2010unary}.}
Using \autoref{thm:pprod-to-sssum}, we can reduce a $\pprod$ instance $(a_1, \dots, a_n, t)$ to a $\simulsub$ instance containing $k$ $\ssum$ instances $(e_{1i}, \dots, e_{ni}, t_i), \forall i \in [k]$ where $k \leq \log(t)$. The space required for the $\simulsub$ instance is the number of bits in $e_{ij}, t_j$. We know that $a_i = \prod_{j=1}^k p_i^{e_{ij}} \implies 2^{\sum_{j=1}^k \log(e_{ij})} \leq 2^{\sum_{j=1}^k e_{ij}} \leq a_i$ because $p_i \geq 2, \forall i \in [k]$. Therefore, we have $\sum_{i,j} \log(e_{i,j}) \leq \sum_{i=1}^n \log(a_i) \leq n\log(t)$. Similarly, $\sum_{i} \log(t_i) \leq \log(t)$. Therefore, the space required for the $\simulsub$ is $O(n\log(t))$. And, if we directly use the low-space algorithm for $\simulsub$ from \cite{kane2010unary}, the total space complexity would become $O((n + \log(nt)) \cdot \log(t))$.

To avoid the $n$-factor in the space complexity, we will not be storing the entire $\simulsub$ instance. Instead, for each summation in the $k$ variate version of \autoref{lem:kane-extension}, we will compute the values of $e_{ij}$ and $t_j$ and discard them after using it. To be precise, for $\overline{g} = (t_1, \dots, t_k)$, we have
    \[c_{\overline{g}} = (-1)^{k} \sum_{\x \in (\F_q^*)^k} f(\x)\cdot\prod_{i=1}^{k} x_i^{q-1-t_i}\]
where $f(\x) = \prod_{i=1}^{n} \left( 1 + \prod_{j=1}^{k} x_j^{e_{ij}}\right)$ and $c_{\overline{g}} = \cf_{\overline{g}}(f(\x))$. The values of $e_{ij}$ and $t_i$ is only required in $f(\x)$ and $\prod_{i=1}^{k} x_i^{q-1-t_i}$ respectively. Since, $e_{ij}$ and $t_{i}$ are the powers of $p_i$ in $a_j$ and $t$ respectively, we can't use pseudo-prime-factorization as this would require us to use $O(n\log(t))$ space to compute a pseudo-prime-factor set. Therefore, we will use naive prime-factorization algorithm that runs in $\tilde{O}(t)$ time which is affordable because we are interested in $\poly(knt)$. 

\paragraph{Choosing the prime $q$.} Observe that the total degree of $f(\x)$ is $\sum_{ij} e_{ij} \leq n\cdot (\sum_i t_i) \le n \log t$ because $0 \leq e_{ij} \leq t_j$ and $\sum_{i} t_i \le \log(t)$. Therefore, the maximum individual degree is bounded by $n\log(t)$. Since, \autoref{lem:kane-extension} requires a prime $q$ that depends on the maximum individual degree of the polynomial, it suffices to work with  $N = \ceil{n\log(t)}$ and $q > N$. Observe that we need to compute the coefficient modulo $q$, therefore, we need to ensure that $q$ does not divide the coefficient. To achieve this, we will use \autoref{lem:kane-extension} for different primes $q \in [N+1, (n+N)^3]$ which contains $\Omega(n+N)^2$ prime. This works because the coefficient can be at most $2^n$, therefore, it will have at most $n$ prime factors. So, at least one prime in the range will not divide the coefficient.

\paragraph{Computing $f(\x)$ and $\prod_{i=1}^{k} x_i^{q-1-t_i}$ using low space.} We will make sure that $t_i$ is the exponent of the $i^{th}$ \textit{smallest} prime factor of $t$. To find an $e_{ij}$, we will first find the $i^{th}$ smallest prime $p_i$ that divides $t$ and then compute the largest power of $p_i$ that divides $a_{j}$. Once, we find $e_{ij}$, we can use it to compute $\prod_{j=1}^{k} x_j^{e_{ij}}$ part of $f(\x)$ and discard it as shown in \autoref{algo:low-space-pprod}. Similarly, we can compute $\prod_{i=1}^{k} x_i^{q-1-t_i}$.

\paragraph{Space and Time complexity.} Observe that \autoref{algo:low-space-pprod} uses only $O(\log(nt))$ space for variables that are used through out the algorithm and reuses $O(\log(t))$ space while computing $t_i, e_{ij}, p_{\ell}$ values. It uses $k\log(nt) = O(\log^2(nt))$ space for $\overline{y}$, therefore, the total space complexity is $O(\log^2(nt))$. Whereas the time complexity is $\poly(nt)$ because each loop runs for $\poly(nt)$ iterations and finding the exponents take $\tilde{O}(t)$ time.

\begin{algorithm}
\caption{Algorithm for solving $\pprod$ using low space}
\label{algo:low-space-pprod}
\vspace{1mm}
\KwIn{A $\pprod$ instance $(a_1, \dots, a_n, t) \in \N^{n+1}$}
\KwOut{Decides whether the $\pprod$ instance has a solution}

$k = 0$\;
\For{each prime $p_i\,|\,t$}{
    $k = k + 1$\;
}
$N = \ceil{n\log(t)}$\;

\For{each prime $q \in [N+1, (n + N)^3]$}{
    $c_g = 1$\;
    \For{each $\overline{y} \in (\F_q^*)^{k}$}{
        $prodx_1 = 1$\;
        \For{$i \in [k]$}{
            Compute $i^{th}$ smallest prime that divides $t$ and find $t_i$\;
            $prodx_1 = prodx_1 * y_i^{q-1-t_i}$\;
            Discard $t_i$\;
        }
        
        $f= 1$\;
        \For{$i \in [n]$}{
            $prodx_2 = 1$\;
            \For{$j \in [k]$}{
                Compute $j^{th}$ smallest prime $p_j$ that divides $t$\;
                Using $p_j$ compute $e_{ij}$ which is the largest integer such that $p_j^{e_{ij}}\,|\,a_i$\;
                $prodx_2 = prodx_2 * y_j^{e_{ij}}$\;
                Discard $p_j$ and $e_{ij}$\;
            }
            $f = f * (1 + prodx_2)$\;
        }
        $c_g = f * prodx_1$\;
    }
    \If{$c_g \neq 0$}{
        \KwRet True\;
    }
}

\KwRet False\;

% $p=2, i = 1$\;
% $val = 1$\;

% \While{$p \leq t$}{
%     \If{$p\,|\,t$}{
%         $prodval = 1$\;
%         \For{$j=1,j\leq n, ++j$}{
%             Compute $e_j^{(i)}$\;
%             $prodval = prodval * y_j^{e_j^{(i)}}$\;
%             Discard $e_j^{(i)}$\;
%         }
%         $val = val * (1+prodval)$\;
%         $i = i + 1$\;
%     }
%     Update $p$ with the next prime\;
% }
% \KwRet $val$\;
\end{algorithm}

\section{An efficient reduction from \texorpdfstring{$\pprod$}{} to \texorpdfstring{$\simulsub$}{}}
\label{sec:pprd-to-ssum}

In this section, we will present a deterministic polynomial time reduction from $\pprod$ to $\simulsub$. In \autoref{sec:rand-algo-pprod}, we have given a pseudo-polynomial time reduction from $\pprod$ to $\simulsub$ by performing prime-factorization of the input $(a_1, \dots, a_n, t)$. The polynomial time reduction also requires to factorize the input, but the factors are not necessarily prime. To be precise, we define pseudo-prime-factorization which can be achieved in polynomial time.

\begin{definition}[Pseudo-prime-factorization] \label{def:pseudo}
A set of integers $\mathcal{P} \subset \N$ is said to be pseudo-prime-factor set of $(a_1, \dots, a_n) \in \N^n$ if 
\begin{enumerate}
    \item the elements of $\mathcal{P}$ are pair-wise coprime, i.e., $\forall p_1,p_2 \in \mathcal{P}, gcd(p_1,p_2) = 1$,
    \item there are only non-trivial factors of $a_i$'s in $\mathcal{P}$, i.e., $\forall p \in \mathcal{P}, \exists i \in [n]$ such that $p\,|\,a_i$,
    \item every $a_i$'s can be uniquely expressed as product of powers of elements of $\mathcal{P}$, i.e., $\forall i \in [n], a_i = \prod_{p \in \mathcal{P}} p^{e_p}, \forall i \in [n]$ where $e_p \geq 0$.
\end{enumerate}
\end{definition}

For a given $(a_1,\hdots, a_n)$, $\mathcal{P}$ may not be unique. A trivial example of a pseudo-prime-factor set of $\mathcal{P}$ for $(a_1, \dots, a_n)$ is the set of all distinct prime factors of $\prod_{i=1}^n a_i$. The following is an important claim which will be used to give a {\em polynomial} time reduction from $\pprod$ to $\simulsub$.

\begin{claim}
\label{claim:size-ppf-set}
For any pseudo-prime-factor set $\mathcal{P}$ of $(a_1, \dots, a_n)$, we have  $|\mathcal{P}| \leq k$ where $k$ is the number of distinct prime factors of $\prod_{i=1}^n a_i$.
\end{claim}

\begin{proof}
The proof is using a simple pigeonhole principle argument. Let $g_1, \dots, g_k$ be the distinct prime factors of $\prod_{i=1}^n a_i$. From the definition of $\mathcal{P}$, we know that $g_1, \dots, g_k$ are the only distinct prime factors of $\prod_{p \in \mathcal{P}}p$. Therefore, if there are more than $k$ numbers in $\mathcal{P}$, then there must exist $p_1, p_2 \in \mathcal{P}$ such that $gcd(p_1, p_2) \neq 1$ which violates pair-wise coprime property of $\mathcal{P}$.
\end{proof}

{\noindent \textsf{$\blacktriangleright$ Constructing $\mathcal{P}$ suffices}.}~~We now show that having a pseudo-prime-factor set $\mathcal{P}$ for $(a_1, \dots, a_n, t)$ helps us to reduce a $\pprod$ instance $(a_1, \dots, a_n, t)$ to $\simulsub$ with number of instances $|\mathcal{P}|$, in polynomial time. Wlog, we can assume that $a_i\,|\,t$ and $a_i, t \leq 2^{m}, \forall i \in [n]$ for some $m$. Trivially, $m \le \log t$. So, using \autoref{claim:size-ppf-set}, we have $|\mathcal{P}| \leq (n+1)\cdot m = \poly(n \log t)$. 

From \autoref{def:pseudo}, we have unique non-negative integers $e_{ij}$ and $t_j$ such that $t = \prod_{j \in |\mathcal{P}|} p_j^{t_j}$ and $a_i = \prod_{j \in |\mathcal{P}|} p_j^{e_{ij}}, \forall i \in [n]$. Since, $a_i\,|\,t$, we have $e_{ij} \leq t_j \leq m, \forall i \in [n], j \in [|\mathcal{P}|]$ and they can be computed in $\poly(m, n)$ time.

Let us consider the $\sssum{|P|}$ instance where the $i^{th}$ $\ssum$ instance is $(e_{1i}, e_{2i}, \dots, e_{ni}, t_i)$. Then, due to unique factorization property of $\mathcal{P}$, the $\pprod$ instance is YES, i.e., $\exists S \in [n]$ such that $\prod_{i \in S} a_i = t$ iff the $\simulsub$ instance with number of instances ${|\mathcal{P}|}$, is a YES.

\subsection{Polynomial time algorithm for computing pseudo-prime-factors}

We will now present a deterministic polynomial time algorithm for computing a pseudo-prime-factor set $\mathcal{P}$ for $(a_1, \dots, a_n)$. We will use the notation $\mathcal{P}(a_1, \dots, a_n)$ to denote a pseudo-prime-factor set for $(a_1, \dots, a_n)$.  Also, let $\calS(a_1,\hdots, a_n)$ be the set of all pseudo-prime-factor sets; this is a finite set. 

The following lemma is a crucial component in \autoref{algo:ppf-set}. We use $a//b$ to denote $a/b^e$ such that $b^{e+1} \nmid a$.

\begin{lemma}
\label{lem:divide-conquer-ppf-set}
Let $(a_1, \dots, a_n)$ be $n$ integers. Then,
\begin{enumerate}
    \item If $a_1$ is coprime with $a_i, \forall i > 1$, then for any $\mathcal{P}(a_2, \dots, a_n) \in \calS(a_2, \dots, a_n)$, $\mathcal{P}(a_2, \dots, a_n) \cup \{a_1\} \in \calS(a_1, \dots, a_n)$.
    
    \item $\mathcal{P}(g, a_1//g, a_2//g, \dots, a_n//g) \in \calS(a_1, \hdots, a_n)$, for given $a_i$, $i \in [n]$ and any factor $g$ of some $a_i$.
\end{enumerate}
\end{lemma}

\begin{proof}

% Let $\mathcal{P}(a_2, \dots, a_n) \cup \{a_1\} = \{p_1, \dots, p_k\}$. We need to show that $\{a_1, p_1, \dots, p_k\}$ satisfies all the properties of a pseudo-prime-factor set. Since, $a_1$ is coprime to $a_2, \dots, a_n$, this implies that it is coprime to all the factors of $\prod_{i=2}^n a_i$. Hence, $a_1$ is coprime to $p_1, \dots, p_k$. Since, $(p_1, \dots, p_k)$ is already a peudo-prime-factor set, the are pair-wise coprime.
The first part of the lemma is trivial. For the second part, let $g$ be a non-trivial factor of some $a_i$ and 
    \[\mathcal{P}\;:=\;\{p_1, \dots, p_k\} \;\in\;\calS(g, a_1//g, a_2//g, \dots, a_n//g)\,,\]
be any pseudo-prime-factor set. Then, $p_i$'s are pair-wise coprime and since each $p_i$ divides either $g$ or $a_i//g$ for some $i \in [n]$, it also divides some $a_i$ because $g$ is a factor of some $a_i$. Also, we have {\em unique} non-negative integers $e_{ip}, e_{gp}$ s.t.~

    \[a_i//g = \prod_{p \in \mathcal{P}} p^{e_{ip}}, \forall i \in [n] \text{ and }g = \prod_{p \in \mathcal{P}} p^{e_{gp}}\;.\]
Combining these equation, we get $a_i = a_i//g * g^{f_{ig}} = \prod_{p \in \mathcal{P}} p^{e_{ip} + e_{gp}*f_{ig}}$. Here $f_{ig}$ is the maximum power of $g$ that divides $a_i$.    
Therefore, $\{p_1, \dots, p_k\}$ is also a pseudo-prime-factor set for $(a_1, \dots, a_n)$. 
\end{proof}

\paragraph{Pre-processing.} Using \autoref{lem:divide-conquer-ppf-set}, \autoref{algo:ppf-set} performs a divide-and-conquer approach to find $\mathcal{P}(a_1, \dots, a_n)$. Observe that we can always remove duplicate elements and $1$'s from the input since it {\em does not} change the pseudo-prime-factors. Also, we can assume without loss of generality that $a_i//a_1 =: a_i, \forall i > 1$ because of the second part in \autoref{lem:divide-conquer-ppf-set}, with $g=a_1$, since it gives us $\calP(a_1,a_2//g, \hdots, a_n//g)$ and we know it suffices to work with these inputs. 

\smallskip
If $a_1$ is coprime to the rest of the $a_i$'s, then the algorithm will recursively call itself on $(a_2, \dots, a_n)$ and combine $\mathcal{P}(a_2, \dots, a_n)$ with $\{a_1\}$. Else, there exist an $i > 1$ such that $gcd(a_1, a_i) \neq 1$. So, the algorithm finds a factor $g$ of $a_1$ using Euclid's GCD algorithm and computes $\mathcal{P}(g, a_1//g, \dots, a_n//g)$. At every step we remove duplicates and $1$'s. Hence, the correctness of \autoref{algo:ppf-set} is immediate assuming it terminates.

To show the termination and time complexity of \autoref{algo:ppf-set}, we will use the {\em `potential function'}  $\potential (I):= \prod_{a \in I} a$, where $I$ is the input and show that at each recursive call, the value of the potential function is halved. Initially, the value of the potential function is $\prod_{i=1}^n a_i$. We also remark that since the algorithm removes duplicates and $1$'s; the potential function can {\em never} increase by the removal step and so it never matters in showing the decreasing nature of $\potential$.
\begin{enumerate}
    \item $a_1$ is corpime to the rest of the $a_i$'s: In this case, the recursive call has input $(a_2, \dots, a_n)$. Since, $a_1 \ge 2$, the value of potential function is 

    \[\potential(a_2,\hdots, a_n) \;=\; \prod_{i=2}^n a_i \;<\; (\prod_{i=1}^n a_i )/2 \;=\; \potential(a_1, \hdots, a_n)/2\,.\] 
    \item $a_1$ shares a common factor with some $a_i$. Let $g = gcd(a_1, a_i) \neq 1$. Since, we have assumed $a_i//a_1 = a_i$, this implies that $a_i$ is not a multiple of $a_1$. This implies that $2 \leq g \leq a_1/2$. Therefore, the new value of potential function is
 
    \begin{align*}
    \potential(g, a_1//g, \hdots, a_n//g)\;&=\; g \prod_{j=1}^n a_j//g \\
        &\;\leq\;  (a_1//g) \times \left( (a_i //g) \times g \right) \times  \prod_{j \in [n]\setminus \{1,i\}} a_j \\
        &\;\leq\; \dfrac{a_1}{g} \cdot \prod_{j=2}^n a_j \\ &\;\leq\; (\prod_{j=1}^n a_j)/2 \;=\; \potential(a_1, \dots, a_n)/2\;.
    \end{align*}
We used the fact that since, $2 \leq g\, |\, a_i$, therefore, $g \times (a_i//g) \leq a_i$.
\end{enumerate}

\paragraph{Time complexity.} In both the cases, the value of the potential function is halved. So, the depth of the recursion tree (in-fact, it is just a line) is at most $\log(\prod_{i=1}^n a_i) \leq m\cdot n$. Also, in each recursive call, the input size is increased at most by one but the integers are still bounded by $2^m$. This implies that input size, for any recurrence call, can be at most $(m+1)\cdot n$. Since there is no branching, the total time complexity is $\poly(m,n) = \poly(\log t, n)$. 

\begin{algorithm}
\caption{Algorithm for Pseudo-prime-factor set}
\label{algo:ppf-set}
\vspace{1mm}
\KwIn{$(a_1, a_2 \dots, a_n) \in \N^n$ which are $m$-bit integers such that $a_i//a_1 = a_i > 1, \forall i > 1$}
\KwOut{Pseudo-prime-factor set $\mathcal{P}$ for $(a_1, a_2, \dots, a_n)$}
\If{n == 0}{
    \KwRet $\emptyset$\;
}
\If{$\exists i > 1$ such that $gcd(a_1, a_i) \neq 1$}{
    $g = gcd(a_1, a_i)$\;
    $I = \{g\}$\;
    \For{$i \in [n]$}{
        $a_i' = a_i//g$\;
        \If{$a_i' \notin I$ and $a_i' \neq 1$}{
            $I = I \cup \{a_i'\}$
        }
    }
    \KwRet $\mathcal{P}(I)$\; 
}
\Else{
    \KwRet $\mathcal{P}(a_2, \dots, a_n) \cup \{a_1\}$\;
}
\end{algorithm}

\section{Extending Theorem 1 and 2 to Unbounded Subset Sum}
\label{sec:unbounded}

In this section, we efficient algorithm for $\ubssum$. The $\ubssum$ is an unbounded variant of $\ssum$\ problem, which is also NP-hard \cite{johnson1985np}.

% \subsection{Extending Theorem 1 and 2 to Unbounded Subset Sum}

\begin{definition}[Unbounded Subset Sum $(\ubssum)$]
Given $(a_1, \dots, a_n, t) \in \Z_{\geq 0}^{n+1}$, the $\ubssum$ problem asks whether there exists $\beta_1, \dots, \beta_n$ such that $\beta_i$ are non-negative integers and $\sum_{i=1}^n \beta_i a_i = t$.
\end{definition}

Similar to the $\ssum$, the $\ubssum$ problem also has a $O(nt)$ dynamic programming algorithm. Interestingly, this problem has a $O(n+ \min_{i} a_i^2)$-time determinstic algorithm \cite{hansen1996testing}. Recently, Bringmann \cite{bringmann2017near} gave an $\tilde{O}(t)$ deterministic algorithm for $\ubssum$. We now define two variants of the $\ubssum$\ problem which is very similar to $\ksssum{k}$ and $\hsssum{k}$.

\begin{problem}[$\ksubssum{k}$]
Given $(a_1, \dots, a_n, t) \in \Z_{\geq 0}^{n+1}$, the $\ksubssum{k}$ problem asks to output all $(\beta_1, \dots, \beta_n)$ where $\beta_i$ are non-negative integers and $\sum_{i=1}^n \beta_i a_i = t$ provided the number of such solutions is at most $k$.
\end{problem}

\begin{problem}[$\hsubssum{k}$]
Given a $\ksubssum{k}$~instance~$(a_1, \dots, a_n, t) \in \Z_{\ge 0}^{n+1}$, $\hsubssum{k}$ asks to output all the hamming weights of the solutions, i.e., $\sum_{i=1}^n \beta_i$.
\end{problem}
\begin{remark}
We want $\vec{a} \cdot \vec{v}=t$, where $\vec{v} \in \Z_{\ge 0}^n$. Similarly, like in the $\ssum$ case (i.e.,~$\vec{v}\in\{0,1\}^n$), we want $|v|_1$, which is exactly the quantity $\sum_{i}\beta_i$, as above. Thus, this definition can be thought as a natural extension of the hamming weight of the solution, in the unbounded regime.
\end{remark}

We will present a deterministic polynomial time reduction from $\ubssum$ to $\mathsf{SimulSubsetSum}$ which will be used latter in this section.

\begin{theorem}[$\ubssum$ reduces to~$\mathsf{SimulSubsetSum}$]
\label{thm:ubssum-to-ssum}
There exists a deterministic polynomial time reduction from $\ubssum$ to $\mathsf{SimulSubsetSum}$.
\end{theorem}
\begin{proof}
Let $(a_1, \dots, a_n, t) \in \Z_{\ge 0}^{n+1}$ be an instance of $\ubssum$. The reduction generates the following $\ssum$ instance $(\underbrace{a_1, 2a_1, 4a_1, \dots, 2^{\gamma}a_1}_{\gamma+1~\text{entries}}, \underbrace{a_2, 2a_2, \dots, 2^{\gamma}a_2}_{\gamma + 1 ~\text{entries}}, \dots,\underbrace{a_n, 2a_n, \dots, 2^{\gamma}a_n}_{\gamma + 1~\text{entries}}, t)$ of size $n(\gamma+1)$ where $\gamma = \floor{\log(t)}$. 

Let $(\beta_1, \dots, \beta_n)$ be a solution to the $\ubssum$ instance, i.e., $\sum_{i=1}^n \beta_i a_i = t$. Since, $\beta_i, a_i, t$ are all non-negative integers, we have $\beta_i \leq t, \forall i \in [n]$. Therefore, $\beta$ is at most $(\gamma + 1)$-bit integer. Let $\beta_i^{(j)}$ be the $j^{th}$ bit of $\beta_i$, then we have

\[
t = \sum_{i=1}^n \beta_i a_i = \sum_{i=1}^n \left(\sum_{j=0}^{\gamma} \beta_i^{(j)}2^{j} \right) \cdot a_i = \sum_{i=1}^n \sum_{j=0}^{\gamma} \beta_{i}^{(j)} \cdot (2^{j}a_i)
\]
which implies that the $\ssum$\ instance also has a solution. Similarly, we can show the reverse direction, i.e.,~if $\ssum$\ instance has a solution, then $\ubssum$ is also has a solution. This concludes the proof. 
\end{proof}

{\noindent $\blacktriangleright$~\textsf{Remark}.}~~
Observe that in \autoref{thm:ubssum-to-ssum}, there is a one-to-one correspondence between the solutions of the $\ubssum$\ and the solutions of the $\mathsf{SimulSubsetSum}$ instance. Therefore, the reduction {\em preserves} the number of solutions. Also, any $T(n,t)$ time algorithm that solve $\ssum$ gives an $T(n\log(t), t)$-time algorithm to solve $\ubssum$.

\medskip
We will now show that the \autoref{thm1:hamming}-\ref{thm2:algo-lowspace} can be extended to, in the $\ubssum$\ regime.

\begin{theorem}
There is an $\tilde{O}(k(n+t))$-time deterministic algorithm for $\hsubssum{k}$.
\end{theorem}

{\em \noindent Proof sketch}.~~The algorithm is almost similar to \autoref{algo:hamming}, except the definitions of the polynomials $f_j(x)$. Here also, we fix $q$ and $\mu$ similarly. We require the exact number of solutions $m (m \le k)$ in \autoref{sec:pf-thm1} (see~\autoref{cl:finding-numer-of-sol}). To do that, define the polynomial $f_0$:

\begin{align*}
    f_0(x)\;:=\; \prod_{i=1}^n\,\left(\dfrac{1}{1-x^{a_j}}\right)\,=\, \left(\prod_{i=1}^n\,\left(1-x^{a_i}\right)\right)^{-1}\,& =\,\left(\prod_{i=1}^n\,\left(1+x^{a_i}+x^{2a_i}+ \hdots\right)\right) \\ \,& =:\,(h_0(x))^{-1}\;.
\end{align*}

In the above, we used the {\em inverse identity} $1/(1-x)\,=\, \sum_{i \ge 0}\,x^i$. Expanding the above, is easy to see that $\cf_{x^t}(f_0(x)) = m$, where $m$ is the {\em exact} number of solutions to the $\ksubssum{k}$. Note that, we can compute $f_0^{-1} = h_0(x) \bmod x^{t+1}$, over $\F_q$ efficiently in $\tilde{O}(k+t)$ time. Finding inverse is {\em easy} and can be done efficiently (see~\cite[Theorem 9.4]{von2013modern}). 

Once, we know $m$, we define $m$ many polynomials $f_j$, for $j \in [m]$, as follows.
\[ f_j(x)\;:=\; \prod_{i=1}^n\,\left(\dfrac{1}{1-\mu^jx^{a_j}}\right)\,=\, \left(\prod_{i=1}^n\,\left(1-\mu^j x^{a_i}\right)\right)^{-1} \,=:\,(h_j(x))^{-1}\]

It is not hard to observe that $\cf_{x^t}(f_j(x)) = \sum_{i \in [\ell]} \lambda_i\cdot \mu^{jw_i}$, where $w_1, \hdots, w_{\ell}$ are the distinct hamming weights with multiplicities $\lambda_1, \hdots, \lambda_{\ell}$ (similar to \autoref{obs:thm1-cf-relation}). To find the coefficients of $f_{j}(x)$, we first compute the coefficients of $h_j(x)$, using \autoref{lem:coeff-extraction}, in $\tilde{O}(k(n+t))$ time and find its inverses, using \cite[Theorem 9.4]{von2013modern}, which can again be done in $\tilde{O}(k(n+t))$ time. Once we have computed the coefficients of $f_j(x)$, the rest proceeds same as \autoref{sec:pf-thm1}.

\begin{theorem}
There is a $\poly(knt)$-time and $O(\log (knt))$-space deterministic algorithm which solves $\ksubssum{k}$.
\end{theorem}
{\em \noindent Proof idea}.~~The algorithm first reduces $\ubssum$ to $\ssum$ using \autoref{thm:ubssum-to-ssum} which preserves the number of solutions but the size of the $\ssum$\ instance is now $n\log t$. Then, it runs \autoref{algo:lowspace} on the $\ssum$\ instance to find all its solutions. From the solutions of the $\ssum$\ instance, it constructs all the solutions of the $\ubssum$\ instance. This gives $\poly(knt)$-time and $O(\log(knt\log t)) = O(\log (knt))$-space algorithm.

\section{Conclusion}
This work introduces some interesting search versions of variants of $\ssum$ problem and gives efficient algorithms for each of them. This opens a variety of questions which require further rigorous investigations.

\begin{enumerate}
% \item In \autoref{remark:randomized-not-helping}, we showed a deterministic~$\tilde{O}(k(k+t))$ algorithm to find both the hamming weights $w_i$ as well as the multiplicities $\lambda_i$. One can also use efficient {\em randomized} factoring algorithms (\cite{berlekamp1970factoring} and \cite[Section~14.6]{von2013modern}), which again takes additional~$\tilde{O}(k^2 \log q)$ time. Can we improve quadratic dependence on $k$ to sub-quadratic (ideally, linear)?
\item Can we improve the time complexity of the \autoref{algo:lowspace}? Because of using \autoref{thm:sparse-interpolation}, the complexity for interpolation is already cubic. Whether some other algebraic (non-algebraic) techniques can improve the time complexity, while keeping it low space, is not at all clear.
\item Can we use these algebraic-number-theoretic techniques, to give a {\em deterministic} $\tilde{O}(n+t)$ algorithm for decision version of $\ssum$?
\item Can we improve \autoref{remark:randomized-not-helping} to find both the hamming weights $w_i$ as well as the multiplicities $\lambda_i$, in $\tilde{O}(k(n+t))$?%. One can also use efficient {\em randomized} factoring algorithms (\cite{berlekamp1970factoring} and \cite[Section~14.6]{von2013modern}), which again takes additional~$\tilde{O}(k^2 \log q)$ time. Can we improve quadratic dependence on $k$ to sub-quadratic (ideally, linear)?
% Currently, it takes $\tilde{O}(k+t)$ deterministic time (\autoref{cl:finding-numer-of-sol}), where $k$ is a given upper bound on the number of solutions, which could be $\exp(n)$ large.
\item Can we improve the complexity of \autoref{thm:randomized-algos} to $\tilde{O}(kn + \prod_{i=1}^k t_i)$? Note that we cannot avoid the $kn$ term in the time complexity because the number of bits in a $\simulsub$ instance is at least $kn$.
    \item What can we say about the hardness of $\simulsub$ with $k$ subset sum instances where $k = \omega(\log(n))$?
    \item Set Cover Conjecture (SeCoCo) \cite{cygan2016problems} states that given $n$ elements and $m$ sets and for any $\epsilon >0$, there is a $k$ such that Set Cover with sets of size at most $k$ {\em cannot} be computed in time $2^{(1-\epsilon)n}\cdot \poly(m)$ time. Can we show that SeCoCo implies there is {\em no} $t^{1-\epsilon} \poly(n)$ time solution to $\pprod$?
\end{enumerate}

\medskip
{\noindent\bf Acknowledgement.}~PD thanks Department of CSE, IIT Kanpur for the hospitality, and acknowledges the support of Google Ph.~D.~Fellowship. MSR acknowledges the support of Prime Minister's Research Fellowship. The authors would like to thank Dr.~Santanu Sarkar (Dept.~of Mathematics, IIT Madras) for introducing us to various algorithms for subset sum. The authors would also like to thank the anonymous reviewers for the suggestions to improve some parts of the presentation.
%%
%% Bibliography
%%

%% Please use bibtex, 

\bibliographystyle{alpha}
\bibliography{bibliography}

\appendix

\section{Generalizing Jin and Wu's technique~\texorpdfstring{\cite{jin2018simple}}{} to different settings} 

\subsection{Revisiting~\texorpdfstring{\cite{jin2018simple}}{} with weighted coefficient}
\label{appendix:coefficient-extraction}
% Since, we are working in the Turing model, we need to analyse the running time of every operation precisely.
% \begin{lemma}
% Let $p$ be a prime. Then, the time to compute 
% \begin{enumerate}
%     \item $a + b \mod{p}$ where $a, b < p$ is $O(\log(p))$.
    
%     \item $ab \mod{p}$ where $a, b < p$ is $\tilde{O}(\log(p))$.
    
%     \item $a \mod{p}$ where $a \in \Z$ is $\tilde{O}(\log(a)\log(p))$.
    
%     \item $a^m \bmod~p$ is $\tilde{O}(\log(am)\log(p))$ (since $a > p$ can be given).
    
%     \item $a^{-1} \mod{p}$ when $0 < a < p$ is $\tilde{O}(\log(a)\log(p))$.
% \end{enumerate}
% \end{lemma}

In \cite[Lemma~5]{jin2018simple}, Jin and Wu established the main lemma which shows that one can compute
\[
A(x)\;:\equiv\;\prod_{i \in [n]}\, (1+x^{a_i})\,\bmod\,\langle x^{t+1}, p \rangle\;,\text{for any prime}\,p\, \in\, [t+1, (n+t)^3]\,,
\]
in $\tilde{O}(t)$ time. Further, choosing a {\em random} $p$,  one can decide nonzeroness of $\cf_{x^t}(A(x))$, with high probability. In this paper, we will work with a more general polynomial 
\[
G(x)\;:\equiv\; \prod_{i\in [n]}\, (1+W^{b}\cdot x^{a_i})\,\bmod\,\langle x^{t+1}, p \rangle\;,
\]
for some integer $W$, not necessarily $1$ and $b \in \Z_{\ge 0}$. Therefore, the details slightly differ. For the completeness, we give the details. But before going into the details, we define some basics of power series and expansion of $\exp$ (respectively~$\ln$), which will be crucially used in the proof of \autoref{lem:coeff-extraction}. In general, we will be working with primes $p$ such that $\log(p) = O(\log(n+t))$, thus $\log(p)$ terms in the complexity can be subsumed in $\tilde{O}$ notation.

\medskip
{\bf \noindent Basic Power series tools.}~We denote $\F[x]$ as the ring of polynomials over a field $\F$, and $\F[[x]]$ denote the ring of formal power series over $\F$ which has elements of the form $\sum_{i \ge 0} a_i x^i$, for $a_i \in \F$. Two important power series over $\Q[[x]]$ are: 
\[
\ln(1+x) \;=\; \sum_{k \ge 1}\, \frac{(-1)^{k-1} x^k}{k}\,,\;\;\text{and}\;\;\exp(x)\;=\;\sum_{k \ge 0}\,\frac{x^k}{k!}\;.
\]
They are inverse to each other and satisfy the basic properties:
\[ \exp\left(\ln\left(1+f(x)\right)\right) = 1+f(x),\;\;\text{and}\;\;\ln\left((1+f(x)) \cdot (1+g(x))\right)=\ln(1+f(x)) + \ln(1+g(x))\;,
\]
for every $f(x), g(x) \in x\Q[x]$ (i.e.,~constant term is $0$). Here is an important lemma to compute $\exp(f(x)) \bmod x^{t+1}$; for details see~\cite{brent1976multiple}; for an alternative proof, see~\cite[Lemma~2]{jin2018simple}.
%Mahesh-check-this (discuss with pranjal)
\begin{lemma}[\cite{brent1976multiple}]\label{lem:exp-lemma}
Given a polynomial $f(x) \in x\F[x]$ of degree at most $t (t < p)$, one can compute a polynomial $g(x) \in \F_p[x]$ in $\tilde{O}(t)$ time such that $g(x) \equiv \exp(f(x))~\bmod~\langle x^{t+1},p \rangle$.
\end{lemma}

Here is the most important lemma, which is an extension of~\cite[Lemma~4]{jin2018simple}, where the authors considered the simplest form. In this paper, we need the extensions for the `robust' usage of this lemma (in \autoref{sec:pf-thm1}).

\begin{lemma}[Coefficient Extraction Lemma] \label{lem:coeff-extraction} Let $A(x) = \prod_{i \in [n]} (1+W^b\cdot x^{a_i})$, for any non-negative integers $a_i, b$ and $W \in \Z$. Then, for a prime $p > t$, one can compute $\cf_{x^r}(A(x))\,\bmod\,p$ for all $0 \le r \le t$, in time $\tilde{O}((n+t \log (Wb)))$.
\end{lemma}

\begin{proof}
Let us define $B(x):=\ln (A(x)) \in \mathbb{Q}[[x]]$. By definition,
\[
B(x) \,=\, \ln \left(\prod_{i \in [n]}\, \left(1+ W^b \cdot x^{a_i}\right)\right)\,=\,\sum_{i \in [n]}\,\ln \left(1 + W^b \cdot x^{a_i}\right)\,=\,\sum_{i \in [n]} \sum_{j=1}^{\infty}\, \frac{(-1)^{j-1}}{j} \cdot W^{jb} \cdot x^{a_i j}\;. 
\]
Let $B_t(x) : = B(x) \bmod \langle x^{t+1},p \rangle$. Define $S_k := \{ i \mid a_i = k\}$. Moreover, let us define
\[
d_{k,j}:= \begin{cases}
\sum_{i \in S_k} W^{jb}, & \;\text{if}\;S_k \,\ne\, \phi\,,\\
0, & \text{otherwise}\,.
\end{cases}
\]
Then, rewriting the above expression, we get

$$B_t(x)\, \equiv \,\sum_{i \in [n]} \sum_{j=1}^{\lfloor t/a_i\rfloor}\, \frac{(-1)^{j-1}}{j} \cdot W^{jb} \cdot x^{a_i j}\,\equiv \,\sum_{k \in [t]} \sum_{j=1}^{\lfloor t/k\rfloor}\, \frac{(-1)^{j-1} \cdot d_{k,j}}{j} \cdot x^{jk}~~\bmod~p\;.$$

Since $p > t$, $j^{-1} \bmod p$ exists, for $j \in [t]$. So we pre-compute all $j^{-1} \bmod p$, which takes total $\tilde{O}(t)$ time. Further, we can pre-compute $|S_k|,\,\forall\,k \in [t]$ in $\tilde{O}((n+t))$ time, just by a linear scan. 

Moreover, computing each $d_{k,j}= |S_k| \cdot W^{jb}$, takes $\tilde{O}(\log (Wbt))$ time, since $j \le t$ (assuming we have computed $|S_k|$). Thus, the total time complexity to compute coefficients of $B_t(x)$ is 
$$ \tilde{O}((n+t)) + \tilde{O}(t ) + \sum_{k \in [t]} \sum_{j \in \lfloor t/k \rfloor}\,\tilde{O}(\log (Wbt) \cdot )\,=\,\tilde{O}((n + t \log (Wb)) )\,.$$
%\[
%\sum_{k \in [t]} \sum_{j \in \lfloor t/k \rfloor}\,\tilde{O}(\log p \cdot \log (Wtb))\,=\,\tilde{O}(\log p \cdot \log (Wtb))\sum_{k \in [t]} \sum_{j \in \lfloor t/k \rfloor}1 =\, \tilde{O}(t \log p \log (Wb))\;.
%\]
In the last, we use that $\sum_{k=1}^{\lfloor t/k \rfloor} 1 = O(t \log t)$, which gives the sum to be $\tilde{O}((n + t \log (Wb)))$ ($\log t$ is absorbed inside the $\tilde{O}$).

The last step is to compute $A(x) \equiv \exp(B_t(x)) \bmod \langle x^{t+1},p \rangle$. Since one can compute $B_t(x)$ in time $\tilde{O}((n+t\log (Wb)))$, using~\autoref{lem:exp-lemma}, one concludes to compute the coefficients of $x^r$ of $A(x)$, $0 \le r \le t$, over $\F_p$ in similar time of~$\tilde{O}((n+t \log (Wb)))$. 
\end{proof}

{\noindent $\blacktriangleright$~\textsf{Remark}.}~~
When $|W|=1$, it is exactly~\cite[Lemma~5]{jin2018simple}. One can also work with $A(x) = \prod_{i \in [n]} (1-W^b\cdot x^{a_i})$; the negative sign does not matter since we can use $\ln(1-x)=-\sum_{i \ge 1}\, - x^i/i$ and the proof goes through.
%Note that maximum coefficient in $A(x)$ can be $(W+1)^n$, so it has at most $nlog(W+1)$ factors. 

\subsection{Fast multivariate polynomial multiplication}
In this section, we will study the time required to compute
\[A(x_1, \dots, x_k)\;:\equiv\;\prod_{i=1}^n\,\left(1 + \prod_{j=1}^k x_j^{a_{ij}}\right)\,\bmod\,\langle x_1^{t_1+1}, x_2^{t_2+1}, \dots, x_k^{t_k+1}, p \rangle\]
for some prime $p$. The case when $k = 1$ has been studied in \cite[Lemma~5]{jin2018simple} where the authors gave an $\tilde{O}(n+t_1)$ time algorithm for $p \in [t_1 + 1, (n+t_1)^3]$. Here we will present the generalisation of this lemma which is used in~\autoref{thm:randomized-algos} using multivariate FFT.

\begin{lemma}[Fast multivariate exponentiation]
\label{lem:multivariate-exponentiation}
Let $\x=(x_1,\hdots,x_k)$ and $f(\x) = \sum_{i=1}^{t_1} f_i(\x) \cdot x_1^{i} \in \F_p[\x]$ where $f_i(\x) \in \F_p[x_2, \dots, x_k]$ such that 
\begin{enumerate}
    \item $f(\x) ~\bmod~\langle x_1, \hdots, x_k \rangle = 0$, i.e.,~the constant term of $f(\x)$ is $0$, and, 
    \item $\deg_{x_j}(f) = t_j$, for positive integers $t_j$.
\end{enumerate}
Then, there is an $\tilde{O}(\prod_{i=1}^{k} (2t_i + 1))$ time deterministic algorithm that computes a polynomial $g(\x) \in \F_p[\x]$ such that $g(\x) \equiv \exp(f(\x))\, \bmod\, \langle x_1^{t_1 + 1}, \dots, x_k^{t_k+1}\rangle$ over $\F_p$.
\end{lemma}

\begin{proof}
Let $g(\x) = \exp(f(\x)) = \sum_{i=0}^{\infty} g_i(x_2, \dots, x_k) \cdot x_1^{i}$, where $g_i \in \F_p[[x_2, \hdots, x_k]]$. Differentiate wrt $x_1$ to get: \[g'(\x) \;:=\; \dfrac{\partial g(\x)}{\partial x_1} \;=\; g(\x) \cdot \dfrac{\partial f(\x)}{\partial x_1}\;.\]
By comparing the coefficients of $x_1^i$ on both sides, we get (over $\F_p$):
\[g_{i} \; \equiv\; i^{-1} \cdot \sum_{j=0}^{i-1} f_{i-j} \cdot g_j\;~\bmod~\langle x_2^{t_2+1}, \hdots, x_k^{t_k+1} \rangle\,,\]
\end{proof}
where $g_0 = 1$. By initializing $g_0 = 1$, the rest $g_i$ to 0 and calling $Compute(0,t_1)$ procedure in \autoref{algo:brent}, we can compute all the coefficients up to $x_1^{t_1}$, in the polynomial $g(\x)~\bmod~\langle x_2^{t_2+1}, \hdots,  x_k^{t_k+1} \rangle$, over $\F_p$.

\begin{algorithm}
\caption{Algorithm for $Compute(\ell,r)$}
\label{algo:brent}
\vspace{1mm}
\KwIn{integers $\ell, r$ and polynomials $f_i, g_i$}
\KwOut{Updated values of $g_i$}

\If{$\ell<r$}{
    $m = \lfloor (\ell + r)/2 \rfloor$\;
    $Compute(\ell, m)$\;
    \For{$i \in \{m+1, \dots, r\}$}{
        $g_i = g_i + i^{-1} \sum_{j = \ell}^{m} (i-j) f_{i-j}g_j\,\mod\,\langle x_2^{t_2+1}, x_3^{t_3+1}, \dots, x_k^{t_k+1}, p \rangle$\;
    }
    $Compute(m+1, r)$\;
}
\KwRet $g_{\ell}, \dots, g_{r}$\;
\end{algorithm}
To speed up this algorithm, we can set $A(\x) = \sum_{i=0}^{r- \ell} if_ix_1^i$ and $B(\x) = \sum_{i=0}^{m-\ell} g_{i + \ell} x_1^i$; here the $f_i$ and $g_{i+\ell}$ have been computed modulo $\langle x_2^{t_2+1}, \hdots,  x_k^{t_k+1} \rangle$ already. Use multidimensional FFT \cite[Chapter~30]{cormen2009introduction} to compute $C(\x) = A(\x)B(\x)$ to speed up the for loop which takes $O(\prod_{i=1}^k (2t_i + 1) \log(\prod_{i=1}^k (2t_i + 1)))$ time.

Observe that $\sum_{j = \ell}^{m} (i-j) f_{i-j}g_j$ is the coefficient of $x_1^{i - \ell}$ in $C(\x)$; importantly $\deg_{x_i}(C) \le 2t_i$, for $i \ge 2$. The extraction of the coefficient of $x_1^i$ in $C(\x)$ for all $i$, $\bmod~\langle x_2^{t_2+1}, x_3^{t_3+1}, \dots, x_k^{t_k+1}\rangle$ can be performed in $O(\prod_{i=1}^k (2t_i + 1))$ time. This is done by traversing through the polynomial and collecting coefficient along with monomials having the same $x_1^i$ term (and there can be at most $\prod_{i=2}^k (2t_i + 1)$ many terms). Thus, the total time complexity of computing $g(\x) \bmod~\langle x_2^{t_2+1}, x_3^{t_3+1}, \dots, x_k^{t_k+1}\rangle$ is 

\[
T(t_1, t_2, \dots, t_k) \;=\; 2T(t_1/2, t_2, \dots, t_k) \,+\, \tilde{O}(\prod_{i=1}^k (2t_i + 1)) \;=\; \tilde{O}(\prod_{i=1}^k (2t_i + 1))\;,
\] as desired. 

\begin{lemma}[Fast logarithm computation]
\label{lem:coefficient-extraction}
Let $A(\x)\,=\,\prod_{i=1}^n\,\left(1 + \prod_{j=1}^k x_j^{a_{ij}}\right)$. Then, there exists an $\tilde{O} (kn + \prod_{i=1}^k t_i)$ time deterministic algorithm that computes $\cf_{\x^{\e}}(\ln(A(\x)))\,\mod\,p$ for all $\e$, such that $\e=(e_1,\hdots, e_k)$ with $e_i \le t_i$. 
\end{lemma}

\begin{proof}
Let us define $B(\x)\,:=\,\ln(A(\x))$. Then,
\begin{align*}
    B(\x) &= \ln \left(\prod_{i=1}^n\,(1 + \prod_{j=1}^k x_j^{a_{ij}})\right) \\
                    &= \sum_{i=1}^n \,\ln\left(1 + \prod_{j=1}^k x_j^{a_{ij}}\right) \\
                    &= \sum_{i=1}^n \sum_{\ell=1}^{\infty}\left( \dfrac{(-1)^{\ell-1}}{\ell}(\prod_{j=1}^k x_j^{a_{ij}})^{\ell}\right)\;.
\end{align*}
Without loss of generality, we can assume that $t_1 \leq t_i, \forall i > 1$. Let $C(\x)\,:= B(\x)\,\mod\,\langle x_1^{t_1+1}, \dots, x_k^{t_k+1}, p \rangle$. Since, we are interested where the individual degree of $x_j$ can be at most $t_j$, the index $\ell$ in the above equation (for a fixed $i$) must satisfy $a_{ij} \cdot \ell \le t_j$ for each $j \in [k]$. This implies $\ell \le t_j/a_{ij}$, for $j \in [k]$. Therefore, define $M_i := min_{j=1}^{k} \lfloor t_j/a_{ij} \rfloor$. Now, one can express $C(\x)$ using $M_i$ since it suffices to look the index $\ell$ till $M_i$ (for a fixed $i$), as argued before.

Importantly, note that the above equation involves $\prod_{j=1}^k x_j^{a_{ij}\ell}$ which has individual degree $>0$, since both $a_{ij}, \ell \ge 1$. Thus, define $T := \{\e= (e_1,\hdots, e_k) \in \Z^{k}\,|\, 1 \leq e_i \leq t_i, \forall i \in [k]\}$. Then,
\begin{align*}
    C(\x) \;&=\; \sum_{i=1}^n \sum_{\ell=1}^{M_{i}}\left( \dfrac{(-1)^{\ell-1}}{\ell}(\prod_{j=1}^k x_j^{a_{ij}})^{\ell}\right) \\
        &=\; \sum_{\overline{e} \in T} \sum_{\ell = 1}^{t_1/e_1} \left( \dfrac{s_{\overline{e}} \times (-1)^{\ell-1}}{\ell} \prod_{j=1}^{k} x_j^{e_i\ell}\right)\;,
\end{align*}
where $s_{\overline{e}} = |\,\{i \in [n]\,|\, (a_{i1}, \dots, a_{ik}) = \overline{e}\}\,|$. Essentially, for a given $s_{\e}$, the quantity computes how many times $a_{ij}$ is equal to $e_j$, for all $j \in [k]$. Using $s_{\e}$, we can interchange the order of the summation as shown above. Moreover, we can pre-compute $s_{\e}$, for all $\e \in T$ in time $O\left(kn + \prod_{i=1}^k t_i\right)$. 

Observe that $\cf_{\x^{\e}}(B(\x)) = \cf_{\x^{\e}}(C(\x))$, for any $(e_1, \dots, e_k) \in T$. Since, $\ell \leq t_1 < p$, $\ell^{-1}$ exists and can be pre-computed in $\tilde{O}(t_1)$.

\paragraph{Time complexity.}
Observe that we have
\begin{align*}
    C(\x) \;&=\; \sum_{\overline{e} \in T} \sum_{\ell = 1}^{t_1/e_1} \left( \dfrac{s_{\overline{e}} \times (-1)^{\ell-1}}{\ell} \prod_{j=1}^{k} x_j^{e_i\ell}\right) \\
    &=\;\sum_{e_2 = 1}^{t_2}\sum_{e_3 = 1}^{t_3}\dots \sum_{e_k=1}^{t_k} \left( \sum_{e_1 = 1}^{t_1} \sum_{\ell = 1}^{t_1/e_1} \left( \dfrac{s_{\overline{e}} \times (-1)^{\ell-1}}{\ell} \prod_{j=1}^{k} x_j^{e_i\ell}\right) \right)
\end{align*}
\\
The time taken to compute all $\cf_{\x^{\e}}(C(\x))$, given $s_{\e}$, is the number of iterations over all $(e_2, \hdots, e_k)$, for $1 \leq e_i \le t_i$, $i >1$ and $\ell \in [t/e_1]$, which is atmost $\sum_{j=1}^{t_1} \floor{t_1/j} \times \prod_{i=2}^{k}t_i = \tilde{O}(\prod_{i=1}^{k} t_i)$, since $\sum_{j=1}^{t_1} t_1/j = O(t_1 \log t_1)$. Thus, the total time is $\tilde{O}(kn + \prod_{i=1}^k t_i)$.
\end{proof}

\subsection{Solving linear recurrence: Tool for \autoref{sec:pf-thm1}} \label{appendix-recurrence-fft}

In this section, we briefly sketch how to speed up the algorithm of computing $E_i$, for $i \in [m]$, using FFT, rather than just going through one by one. \autoref{eq:thm1-relation} gives the following relation:
$$ E_j \equiv j^{-1} \cdot \left(\sum_{i \in [j]}\, (-1)^{j-i-1} E_{i} \cdot P_{j-i}\right)~\bmod~q\;.$$
Here, by $E_j$ (respectively~$P_j$), we mean $E_j(\mu^{w_1},\hdots,\mu^{w_{\ell}})$ (respectively~$P_j$). We can assume that $P_j$'s are already pre-computed and hence contributes to the complexity only once. This calculation is very~similar to \cite[Lemma~2]{jin2018simple}, with a similar relation. But we give the details, for the completeness.

Eventually, once we have computed $P_j$'s, we can use FFT (\autoref{algorithm:fft}) to find $E_j$'s, which eventually gives $T(m) \le \tilde{O}(k(n+t))$. 

To elaborate, in the for-loop 7-8 in \autoref{algorithm:fft}, we want to find $\sum_{i=\ell}^{s} (-1)^{j-i} E_i\cdot P_{j-i}$ for all $j \in \{s+1, \dots, u\}$. To achieve this, we define the polynomials: 
\[
F(x)\,:=\, \sum_{k=0}^{u-\ell} (-1)^{k-1}P_k x^k,~~ \text{and}~~G(x)\,:=\,\sum_{j=0}^{s-\ell} E_{j+\ell}x^j\;.\] 
Note that our $F(x)$ is  {\em different} than used in \cite{jin2018simple}, because of slightly different recurrence relation. We can compute $H(x) = F(x) \cdot G(x)$, in time $\tilde{O}((u-\ell) )$. Observe that $\sum_{i=\ell}^{u}\, (-1)^{j-i-1} P_{j-i} \cdot E_i = \cf_{x^{j-\ell}} (H(x))$ because $(-1)^{j-i-1}P_{j-i} = \cf_{x^{j-i}}(F(x))$ and $E_i = \cf_{x^{i-\ell}}(G(x))$. Therefore, the inner for loop can be computed in $\tilde{O}((u-\ell) )$ time. 

\paragraph{Final time complexity.} Let $T'(m)$ is the complexity of computing $E_1, \hdots, E_m$ assuming precomputations of $P_j$ and $j^{-1}$. Then,
$$T'(m) \;\le\; 2 T'(m/2) \,+\, \tilde{O}( m ) \;\implies\; T'(m) \,\le\, \tilde{O}(m )\;.$$ 
Therefore, the total complexity of computing $E_1, \hdots, E_m$, is $T(m) = T'(m) + \tilde{O}(k(n+t))$, where $\tilde{O}(k(n+t))$ is for the time for computing $P_j$'s (and $j^{-1}$). Since, $q= O(n+k+t)$ and $m \le k$, we get $T(m) = \tilde{O}(k(n+t))$, as we wanted.

\begin{algorithm}[H]
\caption{Algorithm for computing $E_i$}
\label{algorithm:fft}
\vspace{1mm}
\KwIn{$P_i$, for $i \in [m]$, $q$ and $E_0=1$}
\KwOut{$E_i$ for $i \in [m]$}
Initialize $E_j \leftarrow 0$, for $j \in [m]$\;
\KwRet Compute$(0,m)$\;

\textbf{Procedure }Compute$(\ell,u)~\triangleright$ the values returned by Compute$(\ell, u)$ are the final values $E_{\ell}, \hdots, E_{u}$ are computed\;
\For{$\ell < u$ }{
    $s \leftarrow \lfloor\frac{\ell+u}{2}\rfloor$\\
    Compute$(\ell, s)$\;
    \For{$j \leftarrow s+1, \hdots, u$}{
    $E_j \leftarrow E_j + j^{-1} \cdot \left( \sum_{i=\ell}^s \,(-1)^{j-i}E_{i}\cdot P_{j-i} \right)~\bmod~q $\;
    }
   Compute$(s+1,u)$}
\KwRet $E_{\ell}, \dots, E_u$\;
\end{algorithm}

\section{Algorithms}

\subsection{Trivial solution for \texorpdfstring{$\ksssum{k}$}{}}
\label{appendex:k-SSSUM-trivial}
Bellman's dynamic programming solution for the decision version of~$\ssum$~ is based on the recurrence relation $S((a_1, \dots, a_n), t) = S((a_1, \dots, a_{n-1}), t) \oplus S((a_1, \dots, a_{n-1}), t-a_n)$ where $S((a_1, \dots, a_j), t') = 1 \iff t'$ is a realisable target of $(a_1, \dots, a_j)$. Using this relation, the algorithm needs to store only the values of $S((a_1, \dots, a_{j-1}), t')$ for all $1 \leq t' \leq t$ to compute $S((a_1, \dots, a_j), t'')$ for all $1 \leq t'' \leq t$. So, the time complexity is $O(nt)$ whereas the space complexity is $\Omega(t)$.

For finding all the solutions, we modify the above algorithm by adding a pointer from $S((a_1, \dots, a_j), t')$ to $S((a_1, \dots, a_{j-1}), t-a_{j})$ when both are equal to 1. The same is done for $S((a_1, \dots, a_j), t')$ and $S((a_1, \dots, a_{j-1}), t)$. Apart from these, we also add a pointer from $S({a_i}, a_i)$ to a new node $S(\{\}, 0)$ where $1 \leq i \leq n$. This gives a directed graph of size $O(nt)$ because the out-degree of each node is at most 2. To find all the solutions to the~$\ssum$~, we simply run a modified version of DFS algorithm\footnote{The graph is a directed acyclic one and we can use the algorithm mentioned in \url{https://stackoverflow.com/questions/20262712/enumerating-all-paths-in-a-directed-acyclic-graph}} on this graph to finds all the paths from $S((a_1, \dots, a_n), t)$ to $S(\{\}, 0)$. 

It is evident that if the number of solutions to the~$\ssum$~ instance is $k$, then the number of paths is also $k$. The modified DFS algorithm goes through all the neighbouring vertices of a given vertex, no matter if they are visited or not. Furthermore, any path that starts from $S((a_1, \dots, a_n), t)$ will end at $S(\{\}, 0)$.

Clearly, this algorithm will terminate because the graph is directed acyclic.  The running time and space of the modified DFS algorithm is $O(nk)$ because each path is of length at most $n$ and the algorithm traverses through each path at most twice (the first traversal ends at $S(\{\}, 0)$ which finds the path and the second one is backtracking). Therefore, the total time and space complexity is $O(n(t + k))$.

\subsection{Trivial dynamic algorithm for \texorpdfstring{$\simulsub$}{}} \label{sec:dynamic-simul}
In this section, we sketch a dynamic pseudo-polynomial time algorithm which solves $\simulsub$, with targets $t_1, \hdots, t_k$, in $O(n(t_1+1) \hdots (t_k+1))$ time. This is a direct generalization of Bellman's work~\cite{bellman57}.

The algorithm considers an $n \times (t_1+1) \times \dots \times (t_k+1)$ boolean matrix $M$ and populates it with 0/1 entries. $M[i, j_1, j_2, \dots, j_k]$ has 1 iff the $\simulsub$ instance with $\ell^{th}$ $\ssum$ instance $(a_{1\ell}, a_{2\ell}, \dots,a_{i\ell}, j_i)$ has a solution. Here $i \in [n]$ and $j_i \in [0, t_i]$. Even though we have remarked that wlog $t_i \ge 1, \forall i \in [n]$, we cannot do the same for $a_{ij}$'s. This forces us to look at $j_i \in [0, t_i], \forall i \in [k]$. The algorithm starts by setting $M[1, a_{11}, a_{12}, \dots, a_{1k}] = 1$ and $M[1, j_1, j_2, \dots, j_k] = 0$ for the rest. Then,  using the following recurrence relation, the algorithm populates the rest of the matrix.
    \[
    M[i, j_1, j_2, \dots, j_k] = M[i-1,j_1, j_2, \dots, j_k] + M[i-1, j_1 - a_{i1}, j_2 - a_{i2}, \dots, j_k - a_{ik}]
  \]
i.e., $M[i, j_1, j_2, \dots, j_k]$ is set to 1 iff either $M[i-1,j_1, j_2, \dots, j_k] = 1$ or $M[i-1, j_1 - a_{i1}, j_2 - a_{i2}, \dots, j_k - a_{ik}] = 1$. Since, the size of the matrix is $n(t_1+1)\dots (t_k+1)$, the running time of the algorithm is $O(n(t_1+1)\dots (t_k+1))$.

\section{Dynamic programming approach for \texorpdfstring{$\pprod$}{}}
\label{sec:pprod-dynamic}
In this section, we will briefly discuss the modification to Bellman's dynamic programming approach for $\ssum$ to solve $\pprod$ in deterministic (expected) time $O(nt^{o(1)})$. 

The algorithm starts by removing all $a_i$ that does not divide $t$. Then using the factoring algorithm in \cite{lenstra1992rigorous}, we can factor $t$ into prime factor $p_j$, i.e., $t = \prod_{i \in [k]} p_j^{t_j} = t$, where $k = O(\log(t)/\log\log(t))$. We now compute the DP table $T$ of size $n \times (t_1+1) \times \dots \times (t_k+1)$ such that 
\[
T[i, x_1, \dots, x_k] = 1,~\text{if and only if there exists}~S \in [i],~\text{such that}~\prod_{j \in S} a_j = \prod_{j \in [k]} p_j^{x_j}\;.
\]

Observe that the time complexity of the algorithm is the time taken to populate the DP table with either 1 or 0. Since the size of the DP table is $n \times \prod_{i \in [k]} (1+t_i)$, using the similar analyse mentioned in \autoref{sec:pf-thm1-details}, we can bound the term $\prod_{i \in [k]} (1+t_i)$ by $t^{o(1)}$. Therefore, the total time complexity is $O(nt^{o(1)})$.

\end{document}